\documentclass[aps,pra,groupedaddress,superscriptaddress,twocolumn,notitlepage,bibnotes,nofootinbib]{revtex4-1}

\usepackage[latin1]{inputenc}
\usepackage{amsthm}
\usepackage{amssymb}
\usepackage{amsmath}
\usepackage{bbold}
\usepackage{bbm}
\usepackage{nonfloat}
\usepackage[pdftex]{hyperref}
\usepackage{braket}
\usepackage{dsfont}
\usepackage{mathdots}
\usepackage{mathtools}
\usepackage{enumerate}
\usepackage{csquotes}
\usepackage{stmaryrd}
\usepackage[cal=boondox]{mathalfa}
\usepackage{graphicx}
\usepackage{stackengine}
\usepackage{scalerel}
\usepackage[dvipsnames]{xcolor}
\usepackage{array}
\usepackage{makecell}
\newcolumntype{x}[1]{>{\centering\arraybackslash}p{#1}}
\newtheorem{thm}{Theorem}
\newtheorem*{thm*}{Theorem}
\newtheorem{prop}[thm]{Proposition}
\newtheorem*{prop*}{Proposition}
\newtheorem{lemma}[thm]{Lemma}
\newtheorem*{lemma*}{Lemma}
\newtheorem{cor}[thm]{Corollary}
\newtheorem*{cor*}{Corollary}

\newtheorem*{cj*}{Conjecture}
\newtheorem{Def}[thm]{Definition}
\newtheorem*{Def*}{Definition}

\theoremstyle{definition}

\newcommand{\bb}{\begin{equation}}
\newcommand{\bbb}{\begin{equation*}}
\newcommand{\ee}{\end{equation}}
\newcommand{\eee}{\end{equation*}}
\newcommand{\ba}{\begin{eqnarray}}
\newcommand{\ea}{\end{eqnarray}}
\newcommand*{\coloneqq}{\mathrel{\vcenter{\baselineskip0.5ex \lineskiplimit0pt \hbox{\scriptsize.}\hbox{\scriptsize.}}} =}

\newcommand{\texteq}[1]{\stackrel{\mathclap{\scriptsize \mbox{#1}}}{=}}
\newcommand{\textleq}[1]{\stackrel{\mathclap{\scriptsize \mbox{#1}}}{\leq}}

\newcommand{\textgeq}[1]{\stackrel{\mathclap{\scriptsize \mbox{#1}}}{\geq}}
\newcommand{\ketbra}[1]{\ket{#1}\!\!\bra{#1}}

\newcommand{\G}{\mathrm{\scriptscriptstyle G}}

\newcommand{\id}{\mathds{1}}
\newcommand{\R}{\mathds{R}}

\newcommand{\mF}{\mathcal{F}}
\newcommand{\mG}{\mathcal{G}}
\newcommand{\mV}{\mathcal{V}}

\newcommand{\blam}{\boldsymbol{\lambda}}

\DeclareMathOperator{\Tr}{Tr}

\DeclareMathAlphabet{\pazocal}{OMS}{zplm}{m}{n}

\stackMath

\stackMath

\makeatletter
\newcommand*\rel@kern[1]{\kern#1\dimexpr\macc@kerna}
\newcommand*\widebar[1]{%
  \begingroup
  \def\mathaccent##1##2{%
    \rel@kern{0.8}%
    \overline{\rel@kern{-0.8}\macc@nucleus\rel@kern{0.2}}%
    \rel@kern{-0.2}%
  }%
  \macc@depth\@ne
  \let\math@bgroup\@empty \let\math@egroup\macc@set@skewchar
  \mathsurround\z@ \frozen@everymath{\mathgroup\macc@group\relax}%
  \macc@set@skewchar\relax
  \let\mathaccentV\macc@nested@a
  \macc@nested@a\relax111{#1}%
  \endgroup
}

\definecolor{darkblue}{rgb}{0,0,0.5}
\hypersetup{
    bookmarksnumbered=true, % If Acrobat bookmarks are requested, include section numbers
    unicode=false, % non-Latin characters in Acrobat bookmarks
    pdfstartview={FitH}, % fits the width of the page to the window
    pdftitle={}, % title
    pdfauthor={}, % author
    pdfsubject={}, % subject of the document
    pdfcreator={}, % creator of the document
    pdfproducer={}, % producer of the document
    pdfkeywords={}, % list of keywords
    pdfnewwindow=true, % links in new window
    colorlinks=true, % false: boxed links; true: colored links
    linkcolor=black, % color of internal links
    citecolor=darkblue, % color of links to bibliography
    filecolor=blue, % color of file links
    urlcolor=darkblue % color of external links
}

\begin{document}
\title{Assisted concentration of Gaussian resources}

\begin{abstract}
In spite of their outstanding experimental relevance, Gaussian operations in continuous-variable quantum systems are subjected to fundamental limitations, as it is known that general resources cannot be distilled within the Gaussian paradigm. We show that these limitations can be overcome by considering a collaborative setting where one party increases the amount of local resource with the aid of another party, whose operations are assumed to be Gaussian but are otherwise unrestricted; the two parties can only communicate classically. We show that in single-shot scenarios, unlike in the well-known case of entanglement theory, two-way classical communication does not lead to any improvement over one-way classical communication from the aiding party to the aided party. We then provide a concise general expression for the Gaussian resource of assistance, i.e., the maximum amount of resource that can be obtained when the aiding party holds a purification of the aided party's state, as measured by a general monotone. To demonstrate its usefulness, we apply our result to two important kinds of resources, squeezing and entanglement, and find some simple analytic solutions. In the case of entanglement theory, we are able to find general upper bounds on the regularized Gaussian entanglement of assistance, and to establish additivity for tensor powers of thermal states. This allows us to draw a quantitative and enlightening comparison with the performance of assisted entanglement distillation in the non-Gaussian setting. On the technical side, we develop some variational expressions to handle functions of symplectic eigenvalues that may be of independent interest. Our results suggest further potential for Gaussian operations to play a major role in practical quantum information processing protocols.
\end{abstract}

\author{Ludovico Lami}
\email{ludovico.lami@gmail.com}
\affiliation{School of Mathematical Sciences and Centre for the Mathematics and Theoretical Physics of Quantum Non-Equilibrium Systems, University of Nottingham, University Park, Nottingham NG7 2RD, United Kingdom}
\author{Ryuji Takagi}
\email{rtakagi@mit.edu}
\affiliation{Center for Theoretical Physics and Department of Physics, Massachusetts Institute of Technology, Cambridge, Massachusetts 02139, USA}
%\author{Bartosz Regula}
%\email{bartosz.regula@gmail.com}
%\affiliation{School of Mathematical Sciences and Centre for the Mathematics and Theoretical Physics of Quantum Non-Equilibrium Systems, University of Nottingham, University Park, Nottingham NG7 2RD, United Kingdom}
\author{Gerardo Adesso}
\email{gerardo.adesso@nottingham.ac.uk}
\affiliation{School of Mathematical Sciences and Centre for the Mathematics and Theoretical Physics of Quantum Non-Equilibrium Systems, University of Nottingham, University Park, Nottingham NG7 2RD, United Kingdom}
\date{\today}
\maketitle

\section{Introduction}

In recent decades, quantum optics has become one of the most prominent platforms for the implementation of quantum technologies and the realization of quantum communication experiments~\cite{KLM, Braunstein-review, biblioparis, adesso14, CERF, BUCCO}. On the other hand, exploiting quantum effects and correlations such as entanglement over long distances requires taming the omnipresent noise that would otherwise drive a system's behavior toward classical physics. With the goal of addressing near-term practical goals, it is of paramount importance to identify what we can and cannot do with limited quantum resources. Fortunately, the recently established framework of \emph{quantum resource theories}~\cite{Brandao-Gour, G-resource-theories, Genoni2008quantifying, Takagi2018convex,Albarelli2018wigner,RT-review} allows us to deal with questions of this kind in a systematic way.

In continuous variable quantum systems, Gaussian states and Gaussian operations have proven relatively easy to generate, control, and manipulate, and are therefore ideal candidates for the above theoretical program~\cite{biblioparis, adesso14}. We can therefore consider the special class of \emph{Gaussian quantum resource theories}, whose free states and operations are required to be all Gaussian. Such a concept has been formalized in~\cite{G-resource-theories}. Unfortunately, it has been known for a long time that many fundamental protocols are impossible to realize within this Gaussian setting: these include universal fault-tolerant quantum computation~\cite{KLM}, entanglement distillation~\cite{nogo1, nogo2, nogo3}, and error-correction~\cite{Niset2009} (although some experimental progress has been reported for resource distillation and error-correction under non-Gaussian noises, see, e.g., Refs.~\cite{Heersink2006distillation, Hage2008preparation, Dong2008experimental, Lassen2010ec}).
In fact, the impossibility of certain state transformations turns out to be a general feature of Gaussian quantum resource theories~\cite{G-resource-theories}.

In spite of this lack of universality, in this paper we point out that the Gaussian framework still suffices to accomplish somehow ``easier'' tasks that are however of wide interest in quantum information. Namely, we focus on the assisted paradigm for resource distillation, previously introduced in the context of entanglement theory~\cite{Cohen1998, DiVincenzo1999, Smolin2005, Gour2006} and more recently for quantum coherence~\cite{Chitambar-assisted, bartosz-myself-alex} and quantum thermodynamics~\cite{ben-myself-gerry}. This scenario features two parties, conventionally named Alice and Bob, who hold a bipartite state $\rho_{AB}$ and whose goal is to produce a target state $\sigma_A$ on Alice's system. The difficulty lies in the fact that Alice is restricted to free operations, and $\sigma_A$ will generally contain more resource than the initial reduced state $\rho_A$, thus being out of reach if Alice were to operate alone. However, she can rely on Bob's remote help. Bob is usually assumed to have access to more sophisticated equipment than Alice, meaning that he is not restricted to free operations; however, he can communicate with Alice only classically.
Here the analysis bifurcates, as there are two distinct cases of interest, depending on whether Bob can only send messages to Alice or also Alice-to-Bob communication is possible while the protocol is still running. The first case is usually referred to as the \emph{one-way assisted setting}, while the second is called \emph{collaborative setting}. In all the examples above, it is known that already one-way assistance enables more efficient distillation of most bipartite states~\cite{Cohen1998, DiVincenzo1999, Smolin2005, Gour2006, Chitambar-assisted, bartosz-myself-alex, ben-myself-gerry}. Furthermore, the example of entanglement~\cite{Gour2006} shows that two-way communication can help to distill even more resource. Assisted scenarios for remote preparation and distillation of quantum resources have recently been explored experimentally, e.g., with photonics and superconducting microwave cavity setups~\cite{Wu:17,mivcuda2017experimental,pogorzalek2019secure}.

In this paper we analyze the problem of assisted state transformation in Gaussian quantum resource theories, assuming that Bob is also applying Gaussian operations, though not necessarily free as is the case for Alice. As the term ``distillation'' usually refers to the preparation of ``golden units'' of resource,\footnote{A golden unit is a state from which any other state can be prepared by means of free operations alone. a process that is generally impossible within the Gaussian framework~\cite{G-resource-theories}, the goal of our protocols is typically that of increasing the amount of resource, as measured by some chosen monotone. We prefer to refer to this task as ``resource concentration''.} We introduce in this setting the regularized Gaussian measures of assistance, that quantify assisted resource concentration when Alice and Bob have an asymptotic number of copies of a given state available, and can manipulate all of them at the same time in a Gaussian but coherent fashion.

We start by proving that one-way communication suffices for the most general collaborative protocol in this simplified Gaussian setting, thus reconciling the aforementioned two conceptual branches. Next, we tackle two case studies. First, we look at the simplest Gaussian resource theory, i.e., that of squeezing. For any given multimode Gaussian state on Alice's system we are able to explicitly compute the \emph{maximal Gaussian squeezing of assistance}, i.e., the maximal squeezing that can be induced via assistance by an all-Gaussian Bob, assuming that he holds a purification of that state. The squeezing measure we adopt here is directly related to the minimal quadrature variance along all possible directions in phase space. Here we demonstrate by a simple example the irreversibility of Bob's action: namely, the post-measurement states that Alice obtains are in general inequivalent from the point of view of interconvertibility via free operations.

We then move on to our second case study, the resource theory of Gaussian entanglement. When assistance is taken into account, this setting features three parties: the two sharing the entanglement, Alice and Bob, and the assisting party, now named Charlie. When Charlie is restricted to single-mode Gaussian measurements, the maximal entanglement he can induce between Alice and Bob has previously been studied under the name of \emph{Gaussian localizable entanglement}~\cite{Verstraete2004, Popp2005, Fiurasek2007, Mista2008}. We instead look at the scenario where he is allowed to apply any global Gaussian measurement, and give analytical formulae to compute the maximum \emph{Gaussian entanglement of assistance} thus obtainable between Alice and Bob in the following two theoretically and experimentally relevant cases. (1) When Alice and Bob hold a mode each and their reduced state is a product state, which is the setting relevant for entanglement swapping, we confirm that the previously considered ideal Bell measurements on Charlie's two modes~\cite{Pirandola2006, optimal-G-ent-swapping} are indeed the optimal strategy. This result, to the best of our knowledge unknown before, provides rigorous foundation for the most commonly studied procedures of entanglement swapping in the Gaussian setting. (2) When Alice and Bob's reduced state is in the class of Gaussian least entangled mixed states (GLEMS)~\cite{Adesso2004glemsPRL,Adesso2004glems}. This is an important class of states providing a rigorous lower bound on entanglement for given global and local purities, which can be measured without a costly full tomography~\cite{Ekert2002purity,Filip2002purity}.  
It is also a relevant class for a simple physical situation where Alice, Bob, and Charlie hold three-mode pure state, Charlie gets separated apart from Alice and Bob, and Charlie assists Alice and Bob to gain their entanglement later.  
Finally, using some innovative techniques we derive a general upper bound on the Gaussian entanglement of assistance that is valid for all bipartite Gaussian states and for a wide class of entanglement measures. This latter result allows us to compute the regularized Gaussian entanglement of assistance of a tensor product of identical Gaussian thermal states, and to compare it with its non-Gaussian counterpart~\cite{Smolin2005}. The comparison demonstrates that general non-Gaussian protocols are more efficient than Gaussian ones; the gap turns out to be comparatively large when the local entropies are small, and to reduce to a constant in the opposite limit. To the extent of our knowledge, this is the first instance of a precise quantitative comparison between the efficiency of a nontrivial protocol in the standard and in the Gaussian settings.

The rest of the paper is structured as follows: in Section~\ref{sec preliminaries} we recall the definition of a Gaussian resource theory, and introduce the concept of assisted concentration of resources. In Section~\ref{sec G assisted distillation} we present our first general results on the assisted Gaussian framework. Section~\ref{sec squeezing} deals with our first case study, the resource theory of squeezing. Section~\ref{sec entanglement} is instead devoted to entanglement theory. In Section~\ref{sec non-Gaussian} we address the limitations of the Gaussian setting by considering assistance via non-Gaussian protocols. Finally, in Section~\ref{sec conclusions} we discuss our results, draw some conclusions and point out directions of future research.

\section{Preliminaries} \label{sec preliminaries}

\subsection{Gaussian quantum resource theories} \label{subsec G resource theories}

We start by recalling the basic theory of Gaussian quantum states~\cite{weedbrook12, adesso14, BUCCO}. Consider a system of $n$ harmonic oscillators with canonical operators $x_j,p_k$, customarily arranged as a (column) vector $R\coloneqq (x_1,\ldots, x_n, p_1, \ldots, p_n)^\intercal$. The canonical commutation relations can be written in compact form as $[R,R^\intercal] = i \Omega$, where $\Omega\coloneqq \left( \begin{smallmatrix} 0 & \id \\ - \id & 0 \end{smallmatrix}\right)$ is the standard symplectic form. Gaussian states are by definition (limits of) thermal states of quadratic Hamiltonians, i.e., Hamiltonians of the form $\pazocal{H}=\frac12 R^\intercal H R - s^\intercal H R$ for some $2n\times 2n$ positive semidefinite matrix $H$ and some real vector $s\in\R^{2n}$.
A Gaussian state $\rho \coloneqq \rho_{\G}[V,t]$ is uniquely described by its mean or displacement vector $t \coloneqq \Tr[\rho R] \in \R^{2n}$ and its quantum covariance matrix (QCM) $V\coloneqq \Tr \left[\rho\, \{R-t,(R-t)^\intercal\}\right]$, which is a $2n\times 2n$ real symmetric matrix. It turns out that legitimate QCMs of (Gaussian) states are exactly those matrices $V$ that satisfy the Robertson--Schr\"odinger uncertainty principle~\cite{simon94}:
\bb
V \geq i \Omega\, ,
\label{bona fide}
\ee
where the inequality has to be intended in the sense of positive semidefiniteness.

We now move on to Gaussian operations. The simplest example of a Gaussian operation is a symplectic unitary, i.e., a unitary that is generated by a quadratic Hamiltonian. Such a unitary $U$ acts as a symplectic\footnote{A $2n\times 2n$ matrix $S$ is said to be symplectic if it preserves the form $\Omega$, i.e., if $S \Omega S^\intercal = \Omega$.} linear transformation $URU^\dag = SR$ on the vector of canonical operators, while on Gaussian states one has that $U^\dag \rho_{\G}[V,t]\, U = \rho_{\G} [ SVS^\intercal, St]$. Interestingly, symplectic unitaries can be used to bring any Gaussian state into a particularly simple normal form called Williamson's form. The QCM of a state in Williamson's form is simply $V=\left( \begin{smallmatrix} D & \\ & D \end{smallmatrix}\right)$, where $D\geq \id$ is a diagonal matrix whose entries -- called \emph{symplectic eigenvalues} -- depend on $V$ only~\cite{willy}.

Gaussian measurements are represented in the positive operator-valued measure (POVM) formalism by the family of operators $\{\rho_{\G}[\gamma,t]\}_{t\in \R^{2n}}$, where $\gamma$ is a QCM called the seed of the measurement. Note that the normalization condition $\int \frac{d^{2n}t}{(2\pi)^n} \rho_{\G}[\gamma,t] = \id$ is satisfied; here, $\id$ represents the identity acting on the whole Hilbert space. Performing a Gaussian measurement with seed $\gamma_B$ on the $B$ system of a bipartite Gaussian state with QCM $V_{AB}$ yields an outcome distributed normally with covariance matrix $(V_B+\gamma_B)/2$. The post-measurement state on $A$ is again Gaussian, and its QCM $(V_{AB}+\gamma_B)\big/ (V_B+\gamma_B)$ is \emph{independent} of the measurement outcome. Here, for a $2\times 2$ block matrix $M=\left(\begin{smallmatrix} A & X \\ X^\intercal & B \end{smallmatrix}\right)$ the \emph{Schur complement} $M/B$ of $M$ with respect to one of its (invertible) sub-blocks $B$ is defined as $M/B\coloneqq A - XB^{-1}X^\intercal$~\cite{ZHANG}. General non-deterministic Gaussian operations are obtained by appending ancillary Gaussian states, applying symplectic unitaries and performing Gaussian measurements. When a Gaussian operation $\Lambda_{A\to B}$ with input system $A$ and output system $B$ is performed on a Gaussian state, the QCM transforms as 
\bb
\Lambda_{A\to B}\!:\, V_A\longmapsto (\Gamma_{\!AB} + \Sigma_A V_A \Sigma_A)\big/ (\Gamma_{\!A} + \Sigma_A V_A \Sigma_A)\, ,
\label{G operation}
\ee
where $\Gamma_{AB}$ is a QCM pertaining to the joint system $AB$ and characterizing $\Lambda_{A\to B}$, and $\Sigma\coloneqq \left( \begin{smallmatrix} \id & \\ & -\id \end{smallmatrix}\right)$ is the matrix that reverts the sign of the momenta~\cite[Eq.~(10a)]{nogo3}.

We now review the formalism of Gaussian resource theories~\cite{G-resource-theories}. Throughout this paper we will assume that the six postulates proposed in~\cite{G-resource-theories} are satisfied.
Central concepts in the resource theory formalism are a set of free states and a class of free quantum operations. Here, we are interested in a hybrid theory, which takes into account quantum resources restricted to the Gaussian regime. Let $\mF(\blam)$ be the set of free states where $\blam=(\lambda_1,\dots,\lambda_l)$ is a vector variable that specifies the structure of all the ``spatially separated'' subsystems involved, here labelled from $1$ to $l$. For instance, one such variable will be the total number of modes $n_j$ of each subsystem $j$.
Let $\mG_N$ be the set of Gaussian states over $N$ modes. Then, the intersection between the set of free states of the interested resource and the set of Gaussian states, $\mF^{\G}(\blam)=\mF(\blam)\cap\mG_N$ where $N=\sum_j n_j$, defines the set of free Gaussian states, which is our basic object of study. Since in our setting displacement unitaries are always free, $\mF^{\G}(\blam)$ is entirely described by the set of free QCMs $\mV^\G(\blam)\coloneqq \left\{ V:\, \exists\, t\in \R^{2N}:\, \rho_\G[V,t]\in\mF^\G(\blam) \right\}$, in formula $\mF^{\G}(\blam) = \left\{ \rho_\G[V,t]:\, V\in \mV^\G(\blam),\, t\in \R^{2N} \right\}$.

The second ingredient of a Gaussian resource theory is a class of Gaussian quantum channels that are considered to be free. This set can be a priori arbitrary, but we will always require that a free Gaussian operation does not transform free Gaussian states into non-free ones. It will always be clear from the context what particular class of free operations we are looking at.

Before we define the quantities we will study here, we need to fix some terminology. Given a Gaussian quantum resource theory, a \emph{Gaussian monotone} is by definition a function $R$ defined on all Gaussian states and taking on real values, which is non-increasing under the chosen set of Gaussian free operations. This readily implies that any such function must in fact be invariant under any free Gaussian operation whose inverse is also free. Since displacement unitaries are free, we deduce that for a Gaussian state $\rho_\G[V,t]$, the function $E\left( \rho_\G[V,t]\right)$ is in fact only a function of the quantum covariance matrix $V$. With a slight abuse of notation, we will therefore write $E(V)$ instead of $E\left( \rho_\G[V,t]\right)$ in what follows. Combining~\cite[Lemma~S.2]{G-resource-theories} and~\cite[Lemma~1]{Jack2015}, we see that $R$ is continuous as a function of the Gaussian state (with respect to the trace norm) if and only if it is continuous as a function of the covariance matrix. When this happens, $R$ is said to be \emph{continuous}. Another important property is Gaussian convexity: if whenever $\rho=\sum_i p_i \rho_i$ is a Gaussian state that can be written as a convex combination\footnote{We include implicitly also the case where a limit of convex combinations, often written as an integral, has to be taken.} of other Gaussian states $\rho_i$, it holds that $R(\rho) \leq \sum_i p_i R(\rho_i)$, then $R$ is said to be \emph{Gaussian convex}.

\subsection{Assisted resource concentration} \label{subsec assisted}

Here, we are interested in the cooperative scenario featuring two parties, Alice and Bob, whose corresponding quantum systems we denote with $A,B$. The extension (maybe even the purification) of Alice's Gaussian state with covariance matrix $V_A$ is held by Bob, who will help her to implement the desired state transformation $V_A\rightarrow W_A$. We imagine that Bob is restricted to Gaussian operations only but assume that he can implement all Gaussian operations, not only free ones. 
Concerning classical communication between Alice and Bob, it is customary to consider two different settings: (a) there is one-way communication from $B$ to $A$; and (b) we allow for back-and-forth communication at any stage of the protocol.
In the context of entanglement, setting (a) identifies the measure known as entanglement of assistance~\cite{Cohen1998, DiVincenzo1999}, while (b) leads to the definition of the entanglement of collaboration~\cite{Gour2006}. Although the latter is always an upper bound for the former, it has been shown that there can be strict inequality~\cite{Gour2006}. This tells us that in general two-way classical communication must be included into the picture.

Both scenarios (a) and (b) can be easily generalized to the case of arbitrary resources. In the case of (b), the set of operations $A$ and $B$ can perform can be dubbed \emph{Gaussian local free operations and classical communication} and abbreviated as GLFCC. %It is understood that the local free operations on $B$ are all Gaussian operations.
We formally quantify the effectiveness of Gaussian resource concentration in these two settings as follows. 

\begin{Def}
Let $R$ be a Gaussian monotone of a Gaussian resource theory over system $A$. Let the $AB$ bipartite system be in an initial Gaussian state with covariance matrix $V_{AB}$. We assume that $A$ is restricted to free Gaussian operations, while $B$ can perform arbitrary Gaussian operations.
\begin{enumerate}[(a)]
\item If only $B\rightarrow A$ classical communication is allowed, the maximal amount of resource generated by $A$ is given by the \emph{Gaussian resource of one-way collaboration}:
\bb
\ \ R_{c,\leftarrow}^{\G}(V_{AB}) \coloneqq\! \sup_{\gamma_B\geq i\Omega_B}\! R\left( (V_{AB}\!+\! \gamma_B )\big/ (V_B\!+\! \gamma_B ) \right) .
\label{1-way collaboration}
\ee
In particular, if $B$ holds a purification of $A$, we call
\bb
R_a^{\G}(V_A)\coloneqq R_{c,\leftarrow}^{\G}(V_{AB})
\label{assistance}
\ee
the \emph{Gaussian resource of assistance}.
\item Allowing for two-way classical communication instead, we obtain the \emph{Gaussian resource of collaboration}:
\bb
\ R_c^{\G}(V_{AB}) \coloneqq \sup_{\Lambda_{AB\rightarrow A'}\in \mathrm{GLFCC}} R\left( \Lambda_{AB\rightarrow A'}(V_{AB})\right) .
\label{collaboration}
\ee
\end{enumerate}
\end{Def}

\begin{figure}
\centering
\includegraphics[page=1]{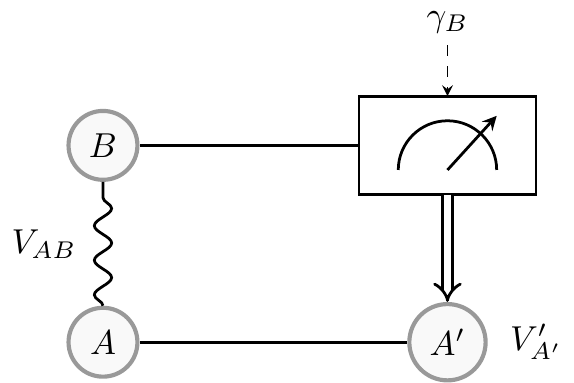}
\caption{A pictorial representation of the one-way collaborative Gaussian concentration paradigm. Alice and Bob's initial Gaussian state has QCM $V_{AB}$. Bob applies a Gaussian measurement with seed $\gamma_B$ to his share of the system and then communicates the outcome to Alice, who then obtains a Gaussian state with QCM $V'_{A'} = (V_{AB} + \gamma_B )\big/ (V_B + \gamma_B )$.}
\label{1-way fig}
\end{figure}

\begin{figure}
\centering
\includegraphics[page=2]{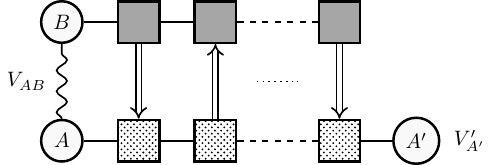}
\caption{A pictorial representation of a two-way collaborative Gaussian concentration scheme. Now, Alice and Bob are allowed to manipulate their initial Gaussian state with QCM $V_{AB}$ using arbitrary Gaussian operations on Bob's side (filled boxes), free Gaussian operations on Alice's (patterned boxes), and two-way classical communication.}
\label{2-way fig}
\end{figure}

To justify~\eqref{1-way collaboration}, remember that the Schur complement on the r.h.s.\ represents the post-measurement QCM obtained when Bob makes a Gaussian measurement with seed $\gamma_B$ on his share of the system.

The setting considered so far involves one-shot manipulation of a single copy of a Gaussian state. From the quantum information theory perspective, it is also natural to investigate the opposite case, i.e., that of asymptotic state manipulation. In this modified scenario, Alice and Bob share $\ell$ copies of a Gaussian state with QCM $V_{AB}$, the global QCM thus taking the form $V_{AB}^{\oplus \ell}$, and may manipulate them jointly, with the help of either one-way or two-way communication. We are thus led to define the following regularized measures.
\begin{enumerate}[(a')]
\item The \emph{regularized Gaussian resource of one-way collaboration}:
\bb
R_{c,\leftarrow}^{\G,\infty}(V_{AB}) \coloneqq \lim_{\ell\to\infty} \frac{1}{\ell} R_{c,\leftarrow}^\G\big( V_{AB}^{\oplus \ell}\big) . 
\label{regularized 1-way collaboration}
\ee
When $B$ holds a purification of $A$ we obtain the \emph{regularized Gaussian resource of assistance}:
\bb
R_a^{\G, \infty}(V_A)\coloneqq \lim_{\ell\to\infty} \frac{1}{\ell} R_a^{\G}\big(V_{AB}^{\oplus \ell}\big) .
\label{regularized assistance}
\ee
\item The \emph{regularized Gaussian resource of collaboration}:
\bb
R_c^{\G, \infty}(V_A)\coloneqq \lim_{\ell\to\infty} \frac{1}{\ell} R_c^{\G}\big(V_{AB}^{\oplus \ell}\big) .
\label{regularized collaboration}
\ee
\end{enumerate}
Since the measures in~\eqref{1-way collaboration}--\eqref{collaboration} are defined via optimizations, it is not difficult to verify that they never decrease under regularization. Naturally, this corresponds to the fact that a possible strategy for Bob is always to measure each copy of the state separately.

In the next section we put order in this zoo of measures, by showing that the Gaussian resource of collaboration and one-way collaboration always coincide. The same is then naturally true for their regularized versions. Introducing and discussing them separately was however no futile exercise, as it helps to keep them conceptually separated. It is indeed important to appreciate that the fact that several of them coincide is really a peculiarity of the Gaussian framework, and will not be the case in other settings. For example, we already mentioned that entanglement of assistance and entanglement of collaboration do not coincide in general~\cite{Gour2006}.

\section{Gaussian assisted resource concentration: first results} \label{sec G assisted distillation}

We start by presenting the main result of this section, i.e., the equality between Gaussian resource of collaboration and one-way collaboration. The operational interpretation of the result is straightforward: one-way communication from Bob (the assisting party) to Alice (the assisted party) is sufficient for optimal resource extraction.

\begin{thm} \label{assistance=collaboration G thm}
The Gaussian resource of collaboration and the Gaussian resource of one-way collaboration coincide, in formula
\bb
R_{c,\leftarrow}^\G (V_{AB}) \equiv R_c^\G(V_{AB})\, .
\label{assistance=collaboration}
\ee
In particular, $R_{c,\leftarrow}^\G$ is a GLFCC monotone.
\end{thm}

In light of its new status, it is important to simplify the computation of the Gaussian resource of one-way collaboration. We now show that the Gaussian measurement on $B$ in~\eqref{1-way collaboration} can always be assumed to have a pure seed.

\begin{prop} \label{assistance optimisation pure prop}
The supremum in~\eqref{1-way collaboration} can be restricted to pure QCMs $\gamma_{B}$ without loss of generality:
\bb
\begin{aligned}
R_c^\G(V_{AB}) &= R_{c,\leftarrow}^{\G}(V_{AB}) \\
&= \sup_{\substack{\\ \gamma_B\geq i\Omega_B \\[.3ex] \text{\emph{$\gamma_B$ pure}}}} R\left( (V_{AB}+ \gamma_B )\big/ (V_B+ \gamma_B ) \right) .
\end{aligned}
\label{eq:assistance optimisation pure}
\ee
\end{prop}

\begin{proof}
The first equality is just~\eqref{assistance=collaboration}. As for the second, consider that for every matrix $\gamma_{B}\geq i\Omega_B$ we can find some pure $\gamma'_B$ such that $\gamma_B \geq \gamma'_B\geq i\Omega_B$. Schur complements are monotonic with respect to the positive semidefinite order, and hence
\bbb
(V_{AB}+ \gamma_B )\big/ (V_B+ \gamma_B )\geq (V_{AB}+ \gamma'_B )\big/ (V_B+ \gamma'_B)\, .
\eee
The proof is concluded by noticing that convex monotones must be decreasing with respect to the positive semidefinite order as increasing the covariance matrix corresponds to acting with random displacements on the underlying quantum state. 
\end{proof}

Until now, our results do not make any special assumption on the bipartite state on $AB$, and thus apply also to the typical experimental setting. The picture becomes a bit clearer, and the assistance protocol more effective, when the global initial state is pure, i.e., when $B$ holds the whole purification of $A$'s reduced state. The analysis of this scenario allows to benchmark the ultimate performances of our Gaussian framework. From the mathematical point of view, it leads to a significant simplification of the optimization~\eqref{assistance}.

\begin{prop}
Let $R$ be a continuous resource monotone. Then the Gaussian resource of assistance of a QCM $V_{A}$ is given by
\bb
R_a^{\G}(V_{A}) = \sup_{\text{\emph{$\tau_A\!\leq\! V_A$ pure QCM}}} R\left( \tau_A \right) 
\label{assistance pure}
\ee
\end{prop}

\begin{proof}
By virtue of~\cite[Proposition~4]{LL-log-det}, for all pure QCMs $V_{AB}$ we have the equality
\begin{align*}
&\overline{\left\{ (V_{AB}+ \gamma_B )\big/ (V_B+ \gamma_B):\ \text{$\gamma_B$ pure QCM} \right\}}\\
&= \left\{ \tau_A: \text{$\tau_A\leq V_A$, $\tau_A$ pure QCM} \right\} ,
\end{align*}
where $\overline{\pazocal{S}}$ denotes the closure of the set $\pazocal{S}$. Since the function $R$ is continuous, the closure does not affect the optimization, hence~\eqref{eq:assistance optimisation pure} becomes~\eqref{assistance pure}.
\end{proof}

The above proof reveals a very intuitive fact: for a given Gaussian state $\rho_\G[V,t]$ on system $A$, and a purification of it on the bipartite system $AB$, every ensemble on $A$ composed by displaced copies of the same Gaussian pure state $\rho_\G[\tau,0]$, with average $\rho_\G[V,t]$, can be obtained by performing a suitable Gaussian measurement on $B$ (or a limit of such measurements). Observe that the possible QCMs $\tau$ here are exactly those satisfying $\tau\leq V$. This is nothing but the Gaussian version of a classic result by Schr\"odinger himself~\cite{schr}. For this reason, in what follows we will often identify a possible strategy by Bob by giving the pure state $\tau$ instead of the measurement $\gamma$ on $B$.

\section{Case study I: squeezing} \label{sec squeezing}

The Gaussian resource theory of squeezing on an $n$-mode system is defined by the set of free QCMs $\mV^{\G}_S(n)\coloneqq \left\{V: V\geq \id\right\}$. This is arguably one of the simplest examples of a Gaussian resource theory. In spite of its simplicity, it already exhibits some interesting features that make it an excellent candidate to investigate and answer some natural questions that arise in the study of Gaussian resources.

An example of such a question is the following: is the optimal measurement on Bob's system in~\eqref{1-way collaboration} -- equivalently, in~\eqref{eq:assistance optimisation pure} -- always independent of the particular monotone $R$ one chooses? In other words, is it always true that one can find an optimal measurement on $B$, whose corresponding induced state on $A$ can be transformed into any other such induced state by means of free operations? Interestingly, this is not the case: the following proposition states that the point achieving the maximum in the optimization~\eqref{assistance pure} can depend in a nontrivial way on the monotone, and hence that the set of states induced on $A$ by measurements on $B$ does not have a maximum with respect to the ordering dictated by convertibility via free operations.

\begin{prop}
Consider the resource theory of Gaussian squeezing over a single system. There exists a QCM $V$ and two pure QCMs $\tau_1, \tau_2\leq V$ such that no pure QCM $\tau$ with $\tau\leq V$ satisfies $\tau\rightarrow \tau_1$ other than $\tau_1$ itself, but $\tau_1\not\rightarrow \tau_2$ by means of free operations.
\label{dependence monotone prop}
\end{prop}

We now move on to the problem of computing the Gaussian squeezing of assistance. As usual, we have to select a monotone to perform the calculation. A commonly employed quantifier in this setting is the so-called \emph{maximal squeezing}, given by~\cite{Fiurasek2004}
\bb
S(V)\coloneqq \max\left\{1,\lambda_{\min}(V)^{-1}\right\} ,
\label{maximal squeezing}
\ee
where $\lambda_{\min}(V)$ denotes the minimal eigenvalue of $V$, i.e., twice the minimal quadrature variance of $\rho_\G[V,t]$ across all possible directions in phase space. Observe that $S$ admits the variational representation $S(V)=\max\{s\geq 1:\,sV\geq \id\}$, and is thus a special case of the general monotone considered in~\cite{G-resource-theories}, defined for an arbitrary Gaussian resource theory by $\kappa(V)=\max\{s\geq 1:\,sV\in \mV^{\G}(\blam)\}$.
We now show that the optimization~\eqref{assistance pure} can be solved explicitly for the resource theory of squeezing when the monotone employed is the maximal squeezing.
This provides a nontrivial example of an explicit computation which gives the optimal assistance.

\begin{thm} \label{S assistance squeezing thm}
Consider the Gaussian resource theory of squeezing over a single system. The Gaussian maximal squeezing of assistance defined in~\eqref{assistance pure} is given by the maximal eigenvalue of the QCM:
\bb
S^{\G}_a(V_{A}) = \lambda_{\max} (V_A)\, .
\label{maximal squeezing a}
\ee
\end{thm}

Since the maximal squeezing $S$ is a special case of the general monotone $\kappa$ defined in~\cite{G-resource-theories}, it obeys the so-called tensorization property, i.e., it holds that $S\big(V^{\oplus \ell}\big)\equiv S(V)$ for any number $\ell$ of copies of the state. It does not come as a surprise that the same is true for the maximal Gaussian squeezing of assistance: $S^\G_a\big(V^{\oplus \ell}\big)\equiv S_a^\G(V)$ for all $\ell$. Clearly, the maximal squeezing is highly non-extensive, and is therefore not apt to capture the asymptotic behavior of the state (one would formally obtain $S_a^{\G,\infty}(V)\equiv 0$). 

The above result can be used to obtain the Gaussian squeezing of assistance with respect to other non-classicality measures that are increasing functions of $S$ for pure states. Examples falling into this class include \emph{all} measures of non-classicality for single-mode states. Let us for instance look at the \emph{entanglement potential} introduced as a computable measure for non-classicality of light~\cite{Asboth2005}. This measures the capability of generating two-mode entanglement with a classical state by a balanced beam splitter. Namely, it is defined by 
\bb
 {\rm EP}(\rho)\coloneqq E_N(\sigma_\rho)=\log_2\|\sigma_\rho^{T_A}\|_1
\ee
where $\sigma_\rho\coloneqq U_{\rm BS}(\rho\otimes\ketbra{0})U_{\rm BS}^{\dagger}$, $U_{\rm BS}$ is the $50:50$ beam splitter, $\ketbra{0}$ is the vacuum mode, and $E_N$ is the negativity of entanglement~\cite{negativity}. 
If one takes the relative entropy measure instead of the negativity, the resulting measure is called \textit{entropic entanglement potential}. The (entropic) entanglement potential can be analytically computed for single-mode Gaussian states as well as photon number states~\cite{Asboth2005}. We then obtain the Gaussian entanglement potential of assistance as follows.
\begin{cor}
 Consider the Gaussian resource theory of squeezing over a single-mode system. Let ${\rm EP}$ be the entanglement potential. Then, we get
 \bb
  {\rm EP}_a^{\G}(V_{A}) = \frac{1}{2}\log_2\lambda_{\max}(V_A).
  \label{EP a}
 \ee
\end{cor}

\begin{proof}
For a pure Gaussian state with QCM $V=R_\phi\,{\rm diag}(\lambda^{-1},\lambda)R_\phi^\dagger$, where $\lambda\geq 1$ and $R_\phi$ is some phase shift operator, the entanglement potential takes the form ${\rm EP}(V)=\frac{1}{2}\log_2\lambda$~\cite{Asboth2005}.
Since this is an increasing function of $\lambda = S(V)$, the pure states that achieve the supremum in Eq.\,(\ref{assistance pure}) for the maximal squeezing and for the entanglement potential coincides, which leads to~\eqref{EP a}. 
\end{proof}

Note that a similar argument can be applied to all the Gaussian resource of assistance measures that are monotonically increasing functions of the $S$ quantifier for pure Gaussian states.

We provide an example where Theorem~\ref{S assistance squeezing thm} is particularly useful; it helps us to observe an interesting difference between finite-dimensional systems and the Gaussian regime in terms of resource conversion between coherence and quantum correlation.
In a finite-dimensional system, the cyclic interconversion between coherence and quantum correlation is possible starting from any coherent state~\cite{hu2016extracting,wu2017experimental}.
More specifically, given a coherent state $\rho_A$, it is always possible to first convert it to a state with nonzero quantum discord $\widetilde{\rho}_{AB}$ by attaching the ancillary state $\ket{0}$ and applying an incoherent operation on $\rho_A\otimes\ketbra{0}$, and convert it back to the original state $\rho_A$ by local quantum incoherent operations and classical communication (LQICC). 
A natural question would be whether one can do the same in the Gaussian regime. 
There, squeezing can be considered the concept analogous to coherence in the sense that it enables to create entanglement together with passive optical elements such as beam splitter.
Let us consider the corresponding cyclic interconversion process where we start from some squeezed state, attach an ancillary state and apply the beam splitter, and try to convert it back to the original state by GLFCC.
Let $\rho_A=\rho_\G[V_A,0]$ be a pure single-mode Gaussian state with QCM $V_A={\rm diag}(\lambda, \lambda^{-1})$, with $\lambda>1$. Set $\widetilde{V}_{AB}={\rm BS}(\tau)\left(V_A\oplus \id_B\right) {\rm BS}(\tau)^\intercal$, where the identity matrix is the QCM of the vacuum state and ${\rm BS}(\tau)$ is the symplectic matrix that represents the action of a beam splitter with transmissivity $\tau$. After a straightforward calculation, one obtains that $\widetilde{V}_A=\tau V_A+(1-\tau)\id$.
Using Theorem~\ref{S assistance squeezing thm}, we get that 
\bb
S^\G_a\big(\widetilde{V}_{A}\big) = \lambda_{\max}\big(\widetilde{V}_A\big) = \tau\lambda+1-\tau \leq \lambda = S(V_A)\,
\ee
where the inequality is strict unless $\tau=1$.
Thus, cyclic interconversion is not always possible in the Gaussian regime.

\section{Case study II: entanglement} \label{sec entanglement}

We now look at the Gaussian resource theory of entanglement over a bipartite system $AB$ of $n_A+n_B$ modes. In this case free QCMs are those that correspond to separable states, i.e.~\cite{Werner01, revisited}
\begin{align*}
\mV^\G_E(n_A, n_B) \coloneqq & \left\{ V_{AB}\!:\, \exists\, \gamma_A,\gamma_B\!:\, V_{AB}\!\geq\! \gamma_A\!\oplus\! \gamma_B\!\geq\! i\Omega_{AB} \right\} \\[0.5ex]
=& \left\{ V_{AB}\!:\, \exists\, \gamma_A\!:\, V_{AB} \geq \gamma_A\!\oplus\! i\Omega_B \geq i \Omega_{AB} \right\} .
\end{align*}
Free transformations are usually taken to be local operations and classical communication (LOCC), which in our setting are required to be also Gaussian (GLOCC)~\cite{giedkemode}.
Let us now consider the problem of Gaussian entanglement generation on a bipartite system $AB$ via assistance by a party ($C$ for Charlie) holding a purification of the state. Instead of selecting a specific monotone, we keep the derivation general and consider any function $E(\gamma_{AB})$ defined on the set of pure QCMs $\gamma_{AB}$ and satisfying a few general properties.
\begin{enumerate}[(i)]
    \item Non-negativity and faithfulness, i.e., $E(\gamma_{AB})\geq 0$ for all pure QCMs $\gamma_{AB}$, with equality if and only if $\gamma_{AB}=\gamma_A\oplus \gamma_B$.
    \item Monotonicity under GLOCC.
    \item Additivity on pure QCMs, i.e., $E(\gamma_{AB}\oplus \tau_{A'B'}) = E(\gamma_{AB}) + E(\tau_{A'B'})$ for all pure QCMs $\gamma_{AB}, \tau_{AB}$.
\end{enumerate}
We obtain the following characterization of such functions.

\begin{lemma} \label{characterization E lemma}
A function $E(\gamma_{AB})$ defined on bipartite pure QCMs $\gamma_{AB}$ obeys conditions (i)--(iii) if and only if it is of the form
\bb
E(\gamma_{AB})=\sum_j f(\nu_j)
\label{characterization E}
\ee
for some non-decreasing function $f:[1,\infty)\to [0,\infty)$ such that $f(1)=0$. Here, $\{\nu_j\}_j$ are the local symplectic eigenvalues of $\gamma_{AB}$, i.e., the symplectic eigenvalues of $\gamma_A$.
\end{lemma}

Notable examples of functions $f$ are: (1) $f(\nu)=s_1(\nu)\coloneqq \frac{\nu+1}{2}\log \left( \frac{\nu+1}{2} \right) - \frac{\nu-1}{2} \log\left( \frac{\nu-1}{2} \right)$, in which case $E$ becomes the standard entanglement entropy~\cite{Horodecki-review}; (2) $f(\nu)=s_2(\nu)\coloneqq \log \nu$, in which case $E$ becomes the R\'enyi-2 entanglement entropy, whose convex roof enjoys a wealth of properties~\cite{AdessoSerafini, GIE, Lami16, LL-log-det}. Incidentally, both $s_1$ and $s_2$ are not only monotonically increasing but also concave.
Also, observe that all choices of $f$ are anyway equivalent when $n_A=n_B=1$, in the sense that there is a monotonic function connecting any two of them.

Let us first consider the simplest case of a product state over a two-mode state. The corresponding QCM will thus be of the form $V_A\oplus V_B$. Clearly product states do not possess any entanglement, but it is possible to produce nonzero entanglement with the help the party holding the purifying system.  

\begin{prop} \label{entanglement product prop}
Let the two-mode QCM $V_{AB}$ represents a product state of the form $V_{AB} = V_A\oplus V_B$. Denote with $a$ and $b$ the symplectic eigenvalues of $V_A$ and $V_B$, respectively. Then, for a monotone $E$ as in~\eqref{characterization E},
\bb
E_a^{\G}(V_{AB})=f\left(\frac{1+ab}{a+b}\right).
\label{entanglement product}
\ee
\end{prop}

Another class of states for which the Gaussian entanglement of assistance can be explicitly calculated is the class of Gaussian least entangled states for given global and local purities (GLEMS)~\cite{Adesso2004glemsPRL, Adesso2004glems}. This class comprises all two-mode states that have at most one symplectic eigenvalue different from unity, i.e., those that admit a single-mode purifying system.

\begin{prop}
Let the two-mode QCM $V_{AB}$ represents a state in the GLEMS class in the standard form $V_{AB} = \left(\begin{smallmatrix} aI & \Delta \\\Delta & bI \end{smallmatrix}\right)$, where $\Delta = {\rm diag}(k_x,k_p)$.
Then, for a monotone $E$ as in~\eqref{characterization E},
\bb
E_a^{\G}(V_{AB}) = f\left(1 + \frac{k_x^2}{ab-k_x^2}\right) .
\label{entanglement GLEMS}
\ee
\label{entanglement GLEMS prop}
\end{prop}

Let us remark once more that \eqref{entanglement product} and \eqref{entanglement GLEMS} represent the maximum average entanglement, as measured by a monotone $E$ obeying requirements (i)--(iii) above, that can be obtained by Alice and Bob with the assistance of Charlie. The setting is the usual: the tripartite state is pure Gaussian, and Charlie can perform a local Gaussian measurement followed by communication of the outcome. While until now we worked in the single-copy setting, we now look at the asymptotic version of the same task, where the three parties possess many independent copies of a fixed state.
The relevant measure is then obtained by taking the regularization of $E_a^\G$, according to the prescription in~\eqref{regularized assistance}:
\bb
E_a^{\G,\infty}(V_{AB}) \coloneqq \lim_{\ell\to\infty} \frac{1}{\ell} E_a^\G\big( V_{AB}^{\oplus \ell}\big) .
\label{regularized Ea}
\ee
As discussed above, it follows from general principles that $E_a^{\G, \infty}(V_{AB})\geq E_a^\G(V_{AB})$ for all bipartite QCMs $V_{AB}$. Unfortunately, we are not yet able to decide whether there can be strict inequality there for some interesting monotone $E$, although we suspect this can happen in general. Note that we do not have a formula to evaluate the limit in~\eqref{regularized Ea}, not even for the simplest (nontrivial) choices of $f$. One should compare this situation with that of standard (non-Gaussian) quantum information: in this case it is known that the regularized entanglement of assistance always equals the minimal local entropy~\cite[Theorem~1]{Smolin2005}; moreover, in~\cite[Example~4]{Smolin2005} an explicit state for which the entanglement of assistance is strictly superadditive is exhibited.

The Gaussian constraints make it difficult to generalize the approach of~\cite{Smolin2005} to the present case. We are however able to compute~\eqref{regularized Ea} in the simplest case of all, i.e., when $V_{AB}=\kappa\id$ is a multiple of the identity -- equivalently, when the state held by Alice and Bob is a product of thermal states with the same mean photon number -- and $f$ is any concave function.  
A by-product of the calculation is a (partially) additive upper bound that holds for \emph{all} QCMs, and that is of interest on its own.

\begin{thm} \label{Ea id thm}
Let the function $E$ be as in~\eqref{characterization E} for some concave $f$. Then for all bipartite QCMs $V_{AB}\geq i\Omega_{AB}$ it holds that
\bb
E_a^{\G}(V_{AB}) \leq E_a^{\G,\infty}(V_{AB}) \leq n\, f\left( \frac{\|V_{AB}\|_\infty^2+1}{2 \|V_{AB}\|_\infty} \right) ,
\label{upper bound Ea}
\ee
where $\|\cdot\|_\infty$ stands for the operator norm. In particular, for all $k\geq 1$
\bb
E_a^{\G}(k \id_{AB}) = E_a^{\G,\infty}(k \id_{AB}) = n\, f\left( \frac{k^2+1}{2k} \right) ,
\label{Ea id}
\ee
where $n=\min\{n_A, n_B\}$ and $n_A$, $n_B$ are the number of modes on the systems. 
\end{thm}

Theorem~\ref{Ea id thm} in particular implies that $E_a^{\G}$ is additive on multiples of the identity. Another particularly important lesson that we can learn from this result is highlighted at the end of Section~\ref{sec non-Gaussian}.

\section{Non-Gaussian operations} \label{sec non-Gaussian}

Finally, let us comment on the possibility of increasing the performance of assisted concentration when the aiding party is allowed to make non-Gaussian operations. Consider the resource theory of optical non-classicality defined on a single system (that is equivalent to the theory of squeezing when it is restricted to the Gaussian states)~\cite{Yadin2018}.
% Suppose $A$ has a single-mode state and $B$ holds the purification of it. 
Recall that the entanglement potential is the amount of entanglement that the state can create by mixing with the vacuum mode in a $50:50$ beam splitter. 
Suppose $A$ and $B$ hold a pure two-mode state $\ket{\psi}_{AB}$ and $C$ holds the vacuum mode $\ket{0}_C$ as depicted in Figure~\ref{localizable fig}. Then, the entanglement potential of assistance for $A$ aided by $B$ is the amount of entanglement created by the beam splitter between the vacuum mode on $C$ and the state on $A$ after the measurement on $B$. 
On the other hand, it is equivalent to the localizable entanglement~\cite{Verstraete2004, Popp2005, Fiurasek2007, Mista2008} induced on $A$ and $C$ from the tripartite pure state $\ket{\phi}_{ABC}\coloneqq U_{\rm BS}^{AC} \ket{\psi}_{AB}\otimes\ket{0}_C$ by a measurement on $B$, where $U_{\rm{BS}}^{AC}$ is the unitary corresponding to a $50:50$ beam splitter on $AC$.
In~\cite{Fiurasek2007}, it was shown that the photon number counting can induce larger localizable entanglement than the optimal Gaussian measurement when $\ket{\psi}_{AB}$ is chosen as a two-mode squeezed vacuum state. This observation directly leads to the fact that entanglement potential of assistance can be increased when non-Gaussian operations are allowed on the aiding system. 

%\begin{figure}[htbp]
%    \centering
%    \includegraphics[scale=0.4]{localizable.pdf}
%    \caption{Equivalence between the entanglement potential of assistance and the localizable entanglement. If one sees that the measurement is performed before the beam splitter, it corresponds to the entanglement potential of assistance while if one sees that measurement is performed after the beam splitter, it corresponds to the localizable entanglement.}
%    \label{localizable entanglement fig}
%\end{figure}

\begin{figure}
\centering
\includegraphics[page=3]{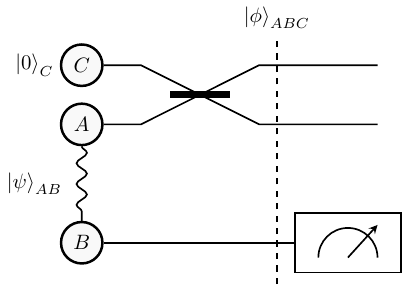}
\caption{Equivalence between the entanglement potential of assistance and the localizable entanglement: if one thinks of the measurement as performed before the beam splitter, then the amount of entanglement between $A$ and $C$ at the output is the entanglement potential of assistance. If instead one sees the measurement as performed after the beam splitter, then the maximum entanglement between $A$ and $C$ is the localizable entanglement of $\ket{\phi}_{ABC}\coloneqq U_{\rm BS}^{AC} \ket{\psi}_{AB}\ket{0}_C$.}
\label{localizable fig}
\end{figure}

Let us now investigate the role of non-Gaussian operations in the assisted preparation of entanglement. We find it particularly instructive to look at the asymptotic case, which can be done for certain states thanks to Theorem~\ref{Ea id thm}. Let $f(\nu)=s_1(\nu)= \frac{\nu+1}{2}\log \left( \frac{\nu+1}{2} \right) - \frac{\nu-1}{2} \log\left( \frac{\nu-1}{2} \right)$ be the function $f$ such that the corresponding $E$ is the standard entanglement entropy. With the help of~\eqref{Ea id} we then see that for all $k>1$ the maximal entanglement that Charlie can induce between Alice and Bob with Gaussian protocols in the asymptotic limit is strictly smaller than the maximal entanglement achievable with \emph{arbitrary} LOCC protocols, which is known to coincide with the minimal local entropy~\cite{Smolin2005}. In formula,
\bb
E_a^{\G,\infty}(k\id_{AB}) = n\, s_1\left( \frac{k^2+1}{2k} \right) < n\, s_1(k) = E_a^\infty (k\id_{AB})\, .
\label{EaG vs Ea}
\ee
Interestingly, the ratio between the r.h.s.\ and the l.h.s.\ of~\eqref{EaG vs Ea} approaches infinity when $k\to 1$, that is, when the local states of Alice and Bob are almost pure. In the opposite limit of $k\to \infty$, i.e., when Alice and Bob both (separately) share a large amount of entanglement with Charlie, the difference between the r.h.s.\ and the l.h.s.\ of~\eqref{EaG vs Ea} tends to a constant $\log 2$, and hence Gaussian protocols perform only slightly worse than general LOCCs.

To the extent of our knowledge, this is one of the few examples of an operational task that can be carried out both in the general and in the Gaussian paradigm, with the corresponding performances being rigorously quantifiable.

\section{Conclusions} \label{sec conclusions}
We considered resource concentration in the collaborative Gaussian setting where one party who only has access to local Gaussian free operations is aided by another party who performs arbitrary local Gaussian operations. We showed that, although two-way classical communications outperform one-way classical communications for the assisted distillation in general, their capability coincide for the Gaussian resource concentration settings considered in this work. 
We in particular analyzed the situation where the aiding party possesses a purification of the aided party, and provided a simplified expression for the general Gaussian resources of assistance. 
We applied this formula to the cases of Gaussian squeezing and two-mode entanglement, and provided analytical solutions for several classes of states and resource measures. 
We finally showed the additivity of the Gaussian entanglement of assistance for product states with respect to a wide class of entanglement measures, providing the analytic formula for the asymptotic scenario. 

Our work illuminates the way of avoiding the well-known {\em no-go} obstacle to using the Gaussian operations -- i.e., impossibility of distilling any Gaussian resources -- by introducing an aiding party that indeed allows the aided party to concentrate the resources. This suggests a further potential for Gaussian operations to be still valuable in various quantum information processing protocols that require Gaussian resource concentration. Our results also provide useful tools for quantitatively assessing the amount of Gaussian resource those protocols can produce.

While we showed the additivity of the Gaussian entanglement of assistance for product states, whether the same holds for other classes of states is not clear, and we leave this for future work. 
Another interesting question to address in the future might be to clarify an operational meaning for the class of entanglement measures considered in this work, which would provide further physical significance of our results. 

\begin{acknowledgements}
We thank Bartosz Regula and Spyros Tserkis for helpful discussions. L.L.\ and G.A.\ acknowledge financial support from the European Research Council (ERC) under the Starting Grant GQCOP (Grant no.~637352). R.T.\ acknowledges the support by NSF, ARO, IARPA, and the Takenaka Scholarship Foundation.
\end{acknowledgements}
\appendix
\widetext

\section{Proof of Theorem~\ref{assistance=collaboration G thm}}
We first need a preliminary lemma that generalizes some of the results obtained in~\cite{nogo3} for the special case of entanglement to the case of general resources. 

\begin{lemma} \label{GLFCC lemma}
Let $\pazocal{O}(A\rightarrow A')$ and $\pazocal{O}(B\rightarrow B')$ be two sets of free Gaussian operations that include all classical noise additions and all displacement unitaries. Given two known bipartite QCMs $V_{AB}$ and $W_{A'B'}$, the transformation $V_{AB}\rightarrow W_{A'B'}$ is possible with general non-deterministic GLFCC if and only if it is possible with a deterministic protocol composed of the following three consecutive stages: (i) local free Gaussian operations; (ii) two-way classical communication; and finally (iii) the application of local displacements.
\end{lemma}

\begin{proof}
One implication is trivial, so we only have to show that any (possibly non-deterministic) GLFCC transformation $V_{AB}\to W_{A'B'}$ can be accomplished by $\text{(i)}\rightarrow\text{(ii)}\rightarrow\text{(iii)}$. At the level of states we will write said transformation as $\rho_{AB}\to \sigma_{A'B'}$, where the QCMs of $\rho_{AB},\sigma_{A'B'}$ are $V_{AB},W_{A'B'}$, respectively (remember that mean vectors are irrelevant as they can always be adjusted via displacement unitaries).
 
Let the GLFCC operation $\Lambda$ that accomplishes the (probabilistic) transformation $\rho_{AB}\to \sigma_{A'B'}$ consist of $N$ rounds of local free Gaussian operations and classical communications. The first round features Alice applying a Gaussian free operation $\Lambda_1$, obtaining a classical measurement outcome modelled by the random variable $T_1$, and communicating it to Bob. In the second round it is Bob instead who applies a Gaussian free operation $\Lambda_2(T_1)$ -- which can depend on Alice's message -- obtains a measurement outcome $T_2$, and sends it to Alice. The protocols proceeds in this way for $N$ rounds, where we can assume $N$ to be even without loss of generality. The operation chosen by either party at the $i$-th round is allowed to depend on all previously exchanged messages. Assume that the state $\sigma_{A'B'}$ is obtained probabilistically conditioned on the sequence of measurement outcomes $(T_1,\ldots, T_N) = (t_1,\ldots, t_N)$. We now describe a protocol to obtain it deterministically via a transformation of the form $\text{(i)}\rightarrow\text{(ii)}\rightarrow\text{(iii)}$.
 
 The protocol is very simple. Step (i) consists of Alice applying operations $\Lambda_1$, then $\Lambda_3(t_1,t_2)$, et cetera, until $\Lambda_{N-1}(t_1,\ldots, t_{N-2})$; at the same time, Bob applies $\Lambda_2(t_1)$, then $\Lambda_4(t_1,t_2,t_3)$, all the way to $\Lambda_{N}(t_1,\ldots, t_{N-1})$. The measurement outcomes obtained are recorded and stored locally by Alice and Bob. Denote with $\omega_{A'B'}$ the state obtained in this way from $\rho_{AB}$, for a specific sequence of measurement outcomes. In step (ii) of the protocol, Alice and Bob reveal to each other the measurement outcomes they obtained. This allows them to reconstruct the mean vector of $\omega_{A'B'}$, which can then be transformed via local displacements in step (iii) so as to match that of $\sigma_{A'B'}$. Call $\widetilde{\omega}_{A'B'}$ the final output state of this protocol. We claim that indeed $\widetilde{\omega}_{A'B'}=\sigma_{A'B'}$.
 
 To see why this is the case, we make use of a crucial property of non-deterministic Gaussian operations: the output QCM does not depend on the measurement outcome that results from the implementation of said operation; in other words, it can be predicted via the rule~\eqref{G operation}. This immediately tells us that the QCM of $\omega_{A'B'}$ (equivalently, the QCM of $\widetilde{\omega}_{A'B'}$) and that of $\sigma_{A'B'}$ are the same. Since also the mean vectors of $\widetilde{\omega}_{A'B'}$ and that of $\sigma_{A'B'}$ coincide, and both states are Gaussian by construction, we see that indeed $\widetilde{\omega}_{A'B'}=\sigma_{A'B'}$, as claimed.

\end{proof}

Now, we are in a position to prove Theorem~\ref{assistance=collaboration G thm}.

\begin{proof}[Proof of Theorem~\ref{assistance=collaboration G thm}]
%Migliorare
The intuition behind the proof is very simple. Since apart from adding (possibly correlated) classical noise the best transformations $A$ and $B$ can apply are local, and a local free operation on $A$ will always decrease the monotone $R$, the optimal protocol consists in fact of a measurement on $B$'s side. In what follows we formalize this intuition appropriately.

Lemma~\ref{GLFCC lemma} entails that the transformation $V_{AB}\rightarrow W_{A'B'}$ is possible with GLFCC if and only if there are QCMs $\Gamma_{AA'}$ and $\Gamma_{BB'}$ such that: (i) the Gaussian operations corresponding to $\Gamma_{AA'}$ and $\Gamma_{BB'}$ via~\eqref{G operation} are free; and (ii) $W_{A'B'} \geq \left(\Gamma_{AA'}\oplus \Gamma_{BB'} + \Sigma_{AB} V_{AB} \Sigma_{AB} \right) \big/ \left(\Gamma_{A}\oplus \Gamma_{B} + \Sigma_{AB} V_{AB} \Sigma_{AB} \right)$. From (ii) we deduce that
\begin{align*}
W_{A'} &\geq \Pi_{A'} \left( \left(\Gamma_{AA'}\oplus \Gamma_{BB'} + \Sigma_{AB} V_{AB} \Sigma_{AB} \right) \big/ \left(\Gamma_{A}\oplus \Gamma_{B} + \Sigma_{AB} V_{AB} \Sigma_{AB} \right) \right) \Pi_{A'}^\intercal \\
&= \left(\Gamma_{AA'}\oplus \Gamma_{B} + \Sigma_{AB} V_{AB} \Sigma_{AB} \right) \big/ \left(\Gamma_{A}\oplus \Gamma_{B} + \Sigma_{AB} V_{AB} \Sigma_{AB} \right) \\
&= \left(\Gamma_{AA'} + \Sigma_A V'_{A} \Sigma_A \right) \big/ \left(\Gamma_{A} + \Sigma_A V'_{A} \Sigma_A \right) .
\end{align*}
where $V'_A \coloneqq (V_{AB}+\Sigma_B \Gamma_B \Sigma_B)\big/ (V_B+\Sigma_B \Gamma_B \Sigma_B)$. This implies that the transformation $V_A'\rightarrow W_{A'}$ is possible with a local free operation. Applying the monotonicity properties of $R$, we deduce that $R \left( W_{A'} \right) \leq R\left( V'_A \right)$. This, together with the form of $V_A'$, implies that for any state on $A$ induced by a GLFCC transformation, $B$ can choose a local Gaussian measurement that induces the state on $A$ with larger resource. 
Thus, $R_c^{\G}\leq R_{c,\leftarrow}^{\G}$. Since the converse inequality is trivial, we obtain $R_c^{\G} = R_{c,\leftarrow}^{\G}$.
This is clearly monotone under GLFCC because of the definition of $R_c^{\G}$ and the fact that the concatenation of two successive GLFCC operations is also a GLFCC operation.
\end{proof}

\section{Proof of Theorem~\ref{S assistance squeezing thm} and of Proposition~\ref{dependence monotone prop}}

This appendix is devoted to the proof of the results concerning the Gaussian resource theory of squeezing. We start with some preliminary technical lemmata, then we move on to the proof of Theorem~\ref{S assistance squeezing thm}, and finally we prove Proposition~\ref{dependence monotone prop}.

In what follows, $\braket{\cdot, \cdot}$ denotes the standard Euclidean product on $\R^{2n}$. We start by recalling a well-known fact concerning the group of $2n\times 2n$ orthogonal symplectic matrices, denoted by $\pazocal{K}(n)$.

\begin{lemma}
Let $\ket{x_1}, \ldots, \ket{x_r}\in \R^{2n}$ and $\ket{x'_1}, \ldots, \ket{x'_r} \in\R^{2n}$ be such that:
\begin{enumerate}[(i)]
\item $\braket{v_j | v_k} = \braket{v'_j | v'_k}$ for all $j,k=1,\ldots, r$; and
\item $\braket{v_j | \Omega | v_k} = \braket{v'_j | \Omega | v'_k}$ for all $j,k=1,\ldots, r$.
\end{enumerate}
Then there exists an orthogonal symplectic transformation $K\in\pazocal{K}(n)$ such that $\ket{x'_j}=K \ket{x_j}$ for all $j=1,\ldots, r$. In particular, the action of $\pazocal{K}(n)$ on the Euclidean unit sphere of $\R^{2n}$ is transitive.
\end{lemma}

\begin{proof}
It becomes clear once one complexifies everything. Namely, write $\ket{x_j}=\left( \begin{smallmatrix} \ket{x_{jx}} \\ \ket{x_{jp}} \end{smallmatrix} \right)$, and construct $\ket{z_j}\coloneqq \ket{x_{jx}}+i \ket{x_{jp}}$. Do the same for $\ket{x'_j}$, obtaining $\ket{z'_j}$. Remember that orthogonal symplectic transformations at the level of $\ket{x}$ become unitaries at the level of $\ket{z}$. Condition (i) now becomes $\Re \braket{z_j|z_k} = \Re \braket{z'_j|z'_k}$, and (ii) can be cast as $\Im \braket{z_j|z_k} = \Im \braket{z'_j|z'_k}$, so together they imply that $\braket{z_j|z_k} = \braket{z'_j|z'_k}$ as complex numbers. This implies the existence of a unitary $U$ such that $U\ket{z_j}=\ket{z'_j}$ for all $j$, completing the proof.
\end{proof}

\begin{lemma} \label{pure lower bound}
Let $V\geq i\Omega$ be a QCM. Pick an eigenvector $\ket{x}$ of $V$, i.e., $V\ket{x}=\lambda \ket{x}$ for some $\lambda>0$. Then there exists a pure QCM $\tau$ such that:
\begin{enumerate}[(a)]
\item $\tau\leq V$; and
\item $\tau \ket{x}=\lambda \ket{x}$.
\end{enumerate}
\end{lemma}

\begin{proof}
Let us write the symplectic form as 
\bbb
\Omega = \bigoplus_{j=1}^n \begin{pmatrix} 0 & 1 \\ -1 & 0 \end{pmatrix} .
\eee
Since it is clear that we can rotate $V$ with orthogonal symplectics without changing the problem, we use the transitivity of the group $\pazocal{K}(n)$ to assume that $\ket{x}=\ket{1}$. Then we have
\bbb
V = \begin{pmatrix} \lambda & 0 & 0 \\ 0 & a & s^\intercal \\ 0 & s & V' \end{pmatrix} ,
\eee
where $a\in\R$, $s\in \R^{2(n-1)}$ and $V'$ is an $2(n-1)\times 2(n-1)$ matrix. Now we enforce the constraint $V\geq i\Omega$. First of all, we deduce that $a\geq 1/\lambda$. If there is equality, then the first mode of $V$ is already in a pure state, hence $s=0$ and we are done. Otherwise, we can take the Schur complement and obtain
\bbb
0\leq V'-i\Omega - \begin{pmatrix} 0 & s \end{pmatrix} \begin{pmatrix} \lambda & -i \\ i & a \end{pmatrix}^{-1} \begin{pmatrix} 0 \\ s^\intercal \end{pmatrix} = V' - i \Omega - \frac{1}{a - 1/\lambda}\, ss^\intercal\, .
\eee
We infer that $V' - \frac{1}{a - 1/\lambda}\, ss^\intercal$ is a quantum covariance matrix, and in turn that there exists a pure state $\tau'$ such that $\tau'\leq V' - \frac{1}{a - 1/\lambda}\, ss^\intercal$. We can set
\bbb
\tau \coloneqq \begin{pmatrix} \lambda & 0 \\ 0 & 1/\lambda \end{pmatrix} \oplus \tau'\, ,
\eee
and then verify that $V\geq \tau$ and $\tau \ket{x} = \tau \ket{1} = \lambda \ket{1}$ hold true.
\end{proof}

We are now ready to prove Theorem~\ref{S assistance squeezing thm}.

\begin{proof}[Proof of Theorem~\ref{S assistance squeezing thm}]
Start by observing that every pure QCM $\tau\leq V$ is a symplectic matrix, i.e., $\tau \Omega \tau=\Omega$ or equivalently $\Omega^\intercal \tau \Omega = \tau^{-1}$. An immediate consequence of this is that $\det \tau=1$, which implies that $\lambda_{\min}(\tau)\leq 1$ and hence that $S (\tau) = \lambda_{\min}^{-1}(\tau)$. Continuing, we obtain
\bbb
S (\tau) = \lambda_{\min}(\tau)^{-1} \texteq{(1)} \lambda_{\min}(\Omega^\intercal \tau \Omega)^{-1} \texteq{(2)} \lambda_{\min}\left(\tau^{-1}\right)^{-1} \texteq{(3)} \lambda_{\max}(\tau)\, ,
\eee
where we used in order: (1) the invariance of the eigenvalues under congruence by an orthogonal matrix; (2) the fact that $\tau$ is symplectic; and (3) the elementary observation that the eigenvalues of the inverse of a matrix are the reciprocal of the initial eigenvalues. The above identity allows us to rewrite
\bbb
S_a^{\G}(V) = \sup_{\text{$\tau\leq V$ pure QCM}} \lambda_{\max}(\tau)\, .
\eee
From the above expression it is clear that $S_a^{\G}(V)\leq \lambda_{\max}(V)$. In fact, the so-called Weyl's monotonicity principle states that $A\geq B\geq 0$ implies that the ordered eigenvalues satisfy $\lambda_i^\downarrow(A)\geq \lambda_i^\downarrow(B)$.

Conversely, pick a vector $\ket{x}$ such that $V\ket{x}=\lambda_{\max}(V)\ket{x}$. Then Lemma~\ref{pure lower bound} guarantees the existence of a pure QCM $\tau\leq V$ such that $\tau\ket{x} = \lambda_{\max}(V) \ket{x}$, which in particular implies that $\lambda_{\max}(\tau)=\lambda_{\max}(V)$. We deduce immediately that $S^{\G}_a(V)\geq \lambda_{\max}(V)$, completing the proof.
\end{proof}

We conclude by presenting the proof of Proposition~\ref{dependence monotone prop}. Let us start by a preliminary lemma.

\begin{lemma} \label{spectrum monotone lemma}
In the resource theory of Gaussian squeezing, if $V\rightarrow W$ can be accomplished by means of any set of free operations that preserves the set of free states, then 
\bb
\min\left\{\lambda_i^\uparrow(V),1\right\} \leq \min\left\{\lambda_i^\uparrow(W),1\right\}
\ee
for all $i\geq 1$, where $\lambda_i^\uparrow$ denotes the $i$-th smallest eigenvalue.
\end{lemma}

\begin{proof}
Let $\mathcal{V}_S^\G =\{V:\, V\geq \id\}$ be the set of free QCMs. For a symmetric matrix $X$, let $N_-(X)$ denote the number of eigenvalues of $X$ that are strictly negative. For integers $i\geq 1$, define the quantifiers
\begin{align*}
\kappa_i(V) \coloneqq &\ \min\left\{ k\geq 1:\ \exists\, V'\in \mathcal{V}_S^\G:\ N_-(kV-V')\leq i-1\right\} \\
=&\ \min\left\{ k\geq 1:\ N_-(kV-\id)\leq i-1\right\} \\
=&\ \min\left\{ k\geq 1:\ \lambda^\uparrow_i(kV-\id)\geq 0\right\} \\
=&\ \min\left\{ k\geq 1:\ k\lambda^\uparrow_i(V)\geq 1\right\} \\
=&\ \max\left\{\lambda_i^\uparrow(V)^{-1},\, 1 \right\} .
\end{align*}
We have to show that $V\to W$ under free operations implies that $\kappa_i(V)\geq \kappa_i(W)$. In order to do this, let the free operation $\Lambda_{A\to B}$ accomplishing this transformation have input system $A$ and output system $B$. Its action is described by a QCM $\Gamma_{AB}$ via~\eqref{G operation}, thus we have that 
\bb
W_B=(\Gamma_{AB} + \Sigma_A V_A\Sigma_A) \big/ (\Gamma_A+\Sigma_A V_A \Sigma_A)\, .
\label{W=Lambda(V)}
\ee
If $\Lambda_{AB}$ has to map free QCMs into free QCMs, it must be the case that
\bbb
\Lambda_{A\to B}:\id_A\longmapsto (\Gamma_{AB} + \id_A) \big/ (\Gamma_A+\id_A) \geq \id_B\, .
\eee
By the properties of Schur complements~\cite{ZHANG}, this is equivalent to requiring that 
\bb
\Gamma_{AB}\geq (-\id_A)\oplus \id_B\, .
\label{Schur Gamma}
\ee
Now set $k\coloneqq \kappa_i(V)$, and write
\begin{align*}
    N_-(k W_B - \id_B) &\texteq{1} N_-\left( k (\Gamma_{AB} + \Sigma_A V_A\Sigma_A) \big/ (\Gamma_A+\Sigma_A V_A \Sigma_A) - \id_B \right) \\
    &\texteq{2} N_-\left( k \Gamma_{AB} + (k \Sigma_A V_A\Sigma_A)\oplus (-\id_B) \right) - N_-(\Gamma_A + \Sigma_A V_A \Sigma_A) \\
    &\texteq{3} N_-\left( k \Gamma_{AB} + (k \Sigma_A V_A\Sigma_A)\oplus (-\id_B) \right) \\
    &\textleq{4} N_-\left( \Gamma_{AB} + (k \Sigma_A V_A\Sigma_A)\oplus (-\id_B) \right) \\
    &\textleq{5} N_-\left( (-\id_A)\oplus \id_B + (k \Sigma_A V_A\Sigma_A)\oplus (-\id_B) \right) \\
    &= N_-\left( (-\id_A)\oplus 0_B + (k \Sigma_A V_A\Sigma_A)\oplus 0_B \right) \\
    &\texteq{6} N_-\left( k \Sigma_A V_A \Sigma_A - \id_A\right) \\
    &\texteq{7} N_-\left( kV_A - \id_A\right) \\
    &\textleq{8} i-1\, .
\end{align*}
These steps are justified as follows: 1: Is obtained via~\eqref{W=Lambda(V)}. 2: Is an application of the law of additivity of inertia~\cite[Theorem~1.6]{ZHANG}. 3: Follows because $\Gamma_A + \Sigma_AV_A\Sigma_A\geq 0$, as the l.h.s.\ is a sum of positive semidefinite matrices. 4: Descends from two observations: first, that as $\Gamma_{AB}>0$ and $k\geq 1$ we have that $k\Gamma_{AB}\geq \Gamma_{AB}$; and second, that the eigenvalues are monotonic function over the set of symmetric matrices -- this is Weyl's monotonicity principle. 5:  6: Comes from~\eqref{Schur Gamma}. 5: Same reasoning applied to the inequality~\eqref{Schur Gamma}. 6: Is another more elementary instance of the law of additivity of inertia. 7: Follows because the eigenvalues are invariant under conjugation by a unitary matrix such as $\Sigma_A$. Finally, 8: is by definition of $k=\kappa_i(V)$.
Comparing the above inequalities with the variational representation of $\kappa_i$ presented above, we conclude that $\kappa_i(W)\leq k=\kappa_i(V)$, concluding the proof.
\end{proof}

\begin{proof}[Proof of Proposition~\ref{dependence monotone prop}]
Set
\bbb
V \coloneqq \begin{pmatrix} 1/2 & 0 & 0 & 0 \\ 0 & a & b & 0 \\ 0 & b & a & 0 \\ 0 & 0 & 0 & 2 \end{pmatrix} ,\qquad
\Omega = \begin{pmatrix} 0 & 1 & & \\ -1 & 0 & & \\ & & 0 & 1 \\ & & -1 & 0 \end{pmatrix} ,
\eee
with $a>2$ and
\bbb
b\coloneqq \sqrt{(a-2)\left(a-1/2\right)} .
\eee
It is elementary to see that $\frac54<a-b<2$ for all $a>2$. The minimal symplectic eigenvalue of $V$ satisfies $\nu_{\min}(V)=1$, so that $V\geq i\Omega$ is a valid QCM. As for the ordinary eigenvalues, one has $\lambda_{\max}(V)=a+b$ with eigenvector $\ket{x_{1}}\coloneqq \frac{1}{\sqrt2}\left( \ket{2}+\ket{3}\right)$ where we use the notation $\{\ket{1},\ket{2},\ket{3},\ket{4}\}$ for the basis vectors corresponding to the above matrix representation.

We now show that one can find two pure QCMs $\tau_1, \tau_2\leq V$ such that: (a) $\tau_1$ is maximal with respect to the ordering induced by convertibility via free operations, i.e., for every pure QCM $\tau\leq V$ we have that $\tau\rightarrow \tau_1$ via free operations if and only if $\tau=\tau_1$; but (b) $\tau_1\not\rightarrow \tau_2$ by means of free operations.

To construct $\tau_1$ we can take inspiration from Lemma~\ref{pure lower bound}, which says that we can always fix $\lambda_{\max}(\tau_1)=\lambda_{\max}(V)=a+b>2$. Since $a+b$ has multiplicity one in the spectrum of $V$, the only way the inequality $\tau_1\leq V$ can be obeyed is if the eigenvectors corresponding to $a+b$ in $V$ and $\tau_1$ coincide, i.e., if $\tau_1\ket{x_{1}} = (a+b) \ket{x_{1}}$. Since $\tau_1$ is symplectic, we obtain also
\bbb
\tau_1 \Omega \ket{x_{1}} = \Omega \tau_1^{-1} \ket{x_{1}} = \frac{1}{a+b}\, \Omega \ket{x_{1}}\, ,
\eee
i.e., $\Omega \ket{x_{1}}$ is another eigenvector of $\tau_1$. Thus
\bbb
V - (a+b) \ket{x_{1}}\!\!\bra{x_{1}} - \frac{1}{a+b}\, \Omega \ket{x_{1}}\!\!\bra{x_{1}} \Omega^\intercal = \begin{pmatrix} \frac12 \left( 1- \frac{1}{a+b} \right) & 0 & 0 & \frac{1}{2(a+b)} \\ 0 & \frac{a-b}{2} & -\frac{a-b}{2} & 0 \\ 0 & -\frac{a-b}{2} & \frac{a-b}{2} & 0 \\ \frac{1}{2(a+b)} & 0 & 0 & 2- \frac12 \frac{1}{a+b} \end{pmatrix} ,
\eee
which is positive by Lemma~\ref{pure lower bound} and also by direct inspection. Now, since $1/(a+b)<1/2\leq 1$, the second largest eigenvalue $\lambda_2^\downarrow(\tau_1)$ of $\tau_1$ satisfies
\begin{align}
\lambda_2^\downarrow (\tau_1) &\leq \lambda_{\max} \left( V - (a+b) \ket{x_{1}}\!\!\bra{x_{1}} - \frac{1}{a+b}\, \Omega \ket{x_{1}}\!\!\bra{x_{1}} \Omega^\intercal \right) \\
&= \max\left\{ \eta_+,\, \eta_-,\, a-b,\, 0 \right\}. \label{lambda_ineq}
\end{align}
where $\eta_{\pm}=\frac14\left(5-4t\pm\sqrt{9+16t^2}\right)$ and $t=a+b$. 
Since $\lim_{a\to 2}a-b = 2$ and $\lim_{a\to 2}\eta_+ = 1+\sqrt{10}/4<2$, and they are continuous with respect to $a$, there exists $\tilde{a}>2$ such that $a-b> \eta_+$ holds for any $a$ with $2<a\leq\tilde{a}$. 
Thus, if we focus on $V$ with $a$ satisfying $2<a\leq\tilde{a}$,~\eqref{lambda_ineq} always equals $a-b$.  
Moreover, the inequality is saturated if and only if $\ket{x_2}\coloneqq\frac{1}{\sqrt2} \left(\ket{2}-\ket{3}\right)$ is an eigenvector of $\tau_1$ with eigenvalue $a-b$, because of the condition $\tau_1\leq V$.
Noting that $\Omega\ket{x_2}$ is also an eigenvector of $\tau_1$ with eigenvalue $1/(a-b)$, we obtain
\begin{align*}
\tau_1 &= (a+b) \ket{x_{1}}\!\!\bra{x_{1}} + \frac{1}{a+b}\, \Omega \ket{x_{1}}\!\!\bra{x_{1}} \Omega^\intercal +  (a-b) \ket{x_{2}}\!\!\bra{x_{2}} + \frac{1}{a-b}\, \Omega \ket{x_{2}}\!\!\bra{x_{2}} \Omega^\intercal \\
&= \begin{pmatrix} \frac{a}{a^2-b^2} & 0 & 0 & \frac{b}{a^2-b^2} \\ 0 & a & b & 0 \\ 0 & b & a & 0 \\ \frac{b}{a^2-b^2} & 0 & 0 & \frac{a}{a^2-b^2} \end{pmatrix} .
\end{align*}
Until now, we have effectively proved that the above matrix is the only one meeting the following three requirements: (i) $\tau_1\leq V$ must be a pure QCM; (ii) $\lambda_{\max} (\tau_1) =a+b$; and (iii) $\lambda_2^\downarrow (\tau_1)$ is maximal among all values compatible with (i) and (ii). We now see that $\tau_1$ must be maximal with respect to the ordering induced by convertibility via free operations. In fact, any other pure QCM $\tau\leq V$ such that $\tau\rightarrow \tau_1$ via free operations will be such that: (i) is met by construction; (ii) is also obeyed, because using Weyl's monotonicity principle and the monotonicity of $\lambda_{\max}$ under free operations we deduce that $\lambda_{\max} (V) \geq \lambda_{\max}(\tau) \geq  \lambda_{\max}(\tau_1) =\lambda_{\max}(V)$ and thus that $\lambda_{\max}(\tau)=\lambda_{\max}(V)$; and finally (iii) is satisfied by an analogous argument. Since also $\tau$ meets these three requirements, it must be that $\tau=\tau_1$ by the above reasoning.

The proof of our claim is complete once we exhibit another pure QCM $\tau_2\leq V$ such that $\tau_1\not\rightarrow\tau_2$ by means of free operations. Our candidate is
\bbb
\tau_2\coloneqq \begin{pmatrix} 1/2 & 0 & 0 & 0 \\ 0 & 2 & 0 & 0 \\ 0 & 0 & 1/2 & 0 \\ 0 & 0 & 0 & 2 \end{pmatrix} .
\eee
Verifying that $\tau_2\leq V$ is entirely elementary. Also, $\tau_2$ is just the direct sum of two pure QCMs on the two local modes corresponding to the first two and last two rows and columns, thus it is itself a pure QCM. Finally, one has $\lambda_2^\downarrow (\tau_2) = 2 > a-b = \lambda_2^\downarrow(\tau_1)$, which shows that $\tau_1\not\rightarrow\tau_2$ with free operations thanks to Lemma~\ref{spectrum monotone lemma}.
\end{proof}

\section{Proof of Lemma~\ref{characterization E lemma} and Propositions~\ref{entanglement product prop} and~\ref{entanglement GLEMS prop}}

\begin{proof}[Proof of Lemma~\ref{characterization E lemma}]
Let $\nu_1,\ldots, \nu_m$ be the nontrivial local symplectic eigenvalues of $\gamma_{AB}$, i.e., those that are strictly larger than $1$. Since $f(1)=0$, we can always restrict the sum in~\eqref{characterization E} to those only. It is a classic fact that $\gamma_{AB}$ is equivalent, up to local symplectic operations, to the direct sum $\bigoplus_{j=1}^m \gamma(\nu_j)$ of the two-mode squeezed vacuum QCMs $\gamma(\nu_j)$, plus additional irrelevant local vacuum states. Here we set
\bb
\gamma(\nu) \coloneqq \begin{pmatrix} \nu \id_2 & \sqrt{\nu^2-1} \sigma_z \\ \sqrt{\nu^2-1} \sigma_z & \nu \id_2 \end{pmatrix} ,
\ee
and $\sigma_z$ is the third Pauli matrix. Since (ii) implies invariance under local symplectic operations, $E$ must be a function of the set $\{\nu_1,\ldots, \nu_m\}$ only. Additivity (iii) entails that $E(\gamma_{AB})=\sum_{j=1}^m f(\nu_j)$ for some function $f:[1,\infty)\to [0,\infty)$. Remember that one can transform a pure QCM with local symplectic eigenvalues $\{\nu_j\}_j$ into another one with local symplectic eigenvalues $\{\mu_j\}_j$ via GLOCC if and only if $\nu^\downarrow_j \geq \mu_j^\downarrow$ for all $j$, where a superscript $\downarrow$ denotes rearrangement in decreasing order. Invoking once more monotonicity under GLOCC (ii), we see that the function $f$ must necessarily be non-decreasing. The converse statement follows immediately along the same lines.
\end{proof}

We now move on to the proof of Proposition~\ref{entanglement product prop}, which states that the Gaussian entanglement of assistance with a monotone $E$ as in~\eqref{characterization E} of a two-mode Gaussian product state with QCM $V_{AB}=V_A\oplus V_B$ evaluates to $E_a^{\G}(V_{AB})=2f\left(\frac{1+ab}{a+b}\right)$. Here, $a\coloneqq \sqrt{\det V_A}$ and $b\coloneqq \sqrt{\det V_B}$ are the symplectic eigenvalues of $V_A$ and $V_B$, respectively.

\begin{proof}[Proof of Proposition~\ref{entanglement product prop}]
We start by lower bounding $E_a^{\G}(V_{AB})$. Assume that $V_A$ and $V_B$ are both in Williamson's form, so that $V_A=a\id_2$ and $V_B=b\id_2$. Take
\bb
\tau = \tau_c = \begin{pmatrix} c\id_2 & \sqrt{c^2-1} \sigma_z \\ \sqrt{c^2-1}\sigma_z & c\id_2 \end{pmatrix}
\label{tau c}
\ee
as a two-mode squeezed vacuum. The constraint $\tau_c\leq (a\id_2)\oplus(b\id_2)$ boils down to $c\leq \frac{ab+1}{a+b}$, showing that $E_a^{\G}(V_{AB}) \geq f\left( \frac{ab+1}{a+b} \right)$.

To prove the converse, consider a state $\tau\leq V_A \oplus V_B$. Since we are free to apply any local symplectic operation, we can assume that $\tau=\tau_c$ has the form in~\eqref{tau c}. In general, this does not mean that either $V_A$ or $V_B$ are in Williamson's form. However, we have still a bit of freedom in choosing $V_A$ and $V_B$, because
\bbb
(O\oplus \sigma_z O\sigma_z)\, \tau_c\, (O \oplus \sigma_z O \sigma_z)^\intercal = \tau_c
\eee
for all $2\times 2$ orthogonal matrices $O$. We can then effect any transformation of the form
\bbb
V_A\mapsto O V_A O^\intercal \, ,\quad V_B\mapsto \sigma_z O \sigma_z V_B \sigma_z O^\intercal \sigma_z\, .
\eee
We use this freedom to diagonalise $V_A$, i.e., we assume
\bbb
V_A = a \begin{pmatrix} \lambda & \\ & \lambda^{-1} \end{pmatrix}
\eee
for some $\lambda>0$. At this point $V_B$ is fixed. Let its diagonal be
\bbb
\Delta(V_B) = \begin{pmatrix} \mu & \\ & \mu' \end{pmatrix} .
\eee
Since $\det V_B \leq \det \Delta(V_B)$, we see immediately that $\mu\mu'\geq b^2$, which in turn implies that
\bb
\mu+\mu'\geq 2b\, .
\label{min harmonic mean product}
\ee

Now, let us take the Schur complement with respect to the first two rows and columns of the positive semidefinite matrix $V_A \oplus V_B - \tau_c\geq 0$. Using the inversion formula
\bbb
(M+x\id)^{-1} = \frac{\Omega (M+x\id)\Omega^\intercal}{\det (M+x\id)} = \frac{\Omega (M+x\id)\Omega^\intercal}{\det M +x^2 + x \Tr M}\, ,
\eee
valid for $2\times 2$ symmetric matrices $M$, we obtain
\bbb
0\leq V_A - c\id_2 - (c^2-1) \frac{\sigma_z \Omega V_B \Omega^\intercal \sigma_z - c\id}{b^2 + c^2 - c \Tr V_B} = V_A - c\id_2 - (c^2-1) \frac{\sigma_x V_B \sigma_x - c\id}{b^2 + c^2 - c \Tr V_B}\, .
\eee
Taking the diagonal part of this inequality yields
\begin{align}
a\lambda - c &\geq \frac{(c^2-1) (\mu' - c)}{b^2 + c^2 - c (\mu+\mu')}\, ,\\
a\lambda^{-1} - c &\geq \frac{(c^2-1) (\mu - c)}{b^2 + c^2 - c (\mu+\mu')}\, .
\end{align}
Multiplying these two inequalities one gets
\begin{align*}
c^2 + a^2 - ac (\lambda+\lambda^{-1}) &\geq \frac{(c^2-1)^2 \left( \mu\mu' +c^2 - c (\mu+\mu')\right)}{\left(b^2 + c^2 - c (\mu+\mu')\right)^2} \\
&\geq \frac{(c^2-1)^2 \left( b^2 +c^2 - c (\mu+\mu')\right)}{\left(b^2 + c^2 - c (\mu+\mu')\right)^2} \\
&= \frac{(c^2-1)^2 }{b^2 + c^2 - c (\mu+\mu')} \, ,
\end{align*}
where the second inequality comes from the fact that $\mu\mu'\geq b^2$. Using the estimates $\lambda+\lambda^{-1}\geq 2$ and $\mu+\mu'\geq 2b$, we arrive at
\begin{align*}
c^2-1 &\leq \sqrt{\left( c^2 + a^2 - ac (\lambda+\lambda^{-1})\right) \left(b^2 + c^2 - c (\mu+\mu') \right) } \\
&\leq \sqrt{\left( c^2 + a^2 - 2ac\right) \left(b^2 + c^2 - 2b c \right) } \\
&= (a-c)(b-c)\, ,
\end{align*}
which leads once again to $c\leq \frac{ab+1}{a+b}$, as claimed.
\end{proof}

We now prove Proposition~\ref{entanglement GLEMS prop} by noticing that a similar optimization technique employed for computing the Gaussian entanglement of formation developed in~\cite{Wolf2004} can be also used for calculating the Gaussian entanglement of assistance for GLEMS. 
\begin{proof}[Proof of Proposition~\ref{entanglement GLEMS prop}]
Recall that a general two-mode QCM can be brought into the standard form
\bbb
\begin{pmatrix}
 a & k_x \\
 k_x & b
\end{pmatrix}
\oplus
\begin{pmatrix}
 a & k_p \\
 k_p & b
\end{pmatrix}
= C_x\oplus C_p
\eee
by local symplectic transformations. Note that we ordered the elements so that the block structure with respect to $x$ quadrature and $p$ quadrature becomes explicit. Since local symplectic transformations do not alter the entanglement, one can assume that $V_{AB}$ is already brought into the standard form without loss of generality. In general, we say that $V$ is in the $x-p$ separate form if it has a block structure with respect to position part and momentum part. Note that the standard form is a special case of $x-p$ separate form. 
We first show that if the pure state that achieves the maximum in~\eqref{assistance pure} takes the $x-p$ separate form, the optimization can be simplified in an analogous way of~\cite{Wolf2004}, and we finally show that it is the case for GLEMS. 
It is shown in~\cite{Wolf2004} that $n$-mode pure QCM has the specific form; $\begin{pmatrix} 
 X & XY \\
 YX & YXY + X^{-1}
\end{pmatrix}$ where $X>0$ and $Y$ are $n\times n$ real symmetric matrix. 
Thus, a pure QCM is in the $x-p$ separate form if and only if it has the form $\tau(X)\coloneqq X\oplus X^{-1}$.
Let $V_{AB} = C_x\oplus C_p$ and suppose the maximum in~\eqref{assistance pure} is achieved at a pure state in the $x-p$ separate form.
Then, since two-mode entanglement measures are increasing function of the determinant of the QCM for the reduced density matrix, one is to maximize $\det\tau_A(X)$ over $X$ under the constraint 
\ba
C_p^{-1}\leq X \leq C_x.
\label{eq:optimization constraint}
\ea
In fact, the optimal $X$ saturates both inequalities. 
It can be seen by first observing that 
\ba
 \det \tau_A(X)=1+\frac{X_{12}^2}{\det X} = 1+\frac{(X^{-1})_{12}^2}{\det (X^{-1})}. 
 \label{eq:cone det}
\ea
Suppose $X$ satisfies $C_p^{-1}<X\leq C_x$.
Then, one can take another positive matrix $X_\epsilon\coloneqq X-\epsilon I$, $\epsilon>0$, which still satisfies~\eqref{eq:optimization constraint} but has larger entanglement because of~\eqref{eq:cone det}.
Thus, one can increase the entanglement by increasing $\epsilon$ until $C_p^{-1}\leq X_{\epsilon}$ is saturated, i.e., an optimal $X$ always satisfies $\det(X-C_p^{-1})=0$. 
On the other hand, suppose $X$ satisfies $C_p^{-1}\leq X< C_x$, which is equivalent to $C_x^{-1}< X^{-1}\leq C_p$. 
Then, one can always take $X_{\epsilon}\coloneqq \left(X^{-1}-\epsilon I\right)^{-1}>0$, which has larger entanglement because of~\eqref{eq:cone det} while satisfying~\eqref{eq:optimization constraint}. 
By the same argument, $\epsilon$ can be increased until $C_x^{-1}\leq X_{\epsilon}^{-1}$ is saturated, i.e., an optimal $X$ always satisfies $\det(X-C_x)=0$. 
Hence, we conclude that the maximum entanglement is realized when both the inequalities in~\eqref{eq:optimization constraint} are saturated, i.e., $\det(X-C_p^{-1})=\det(X-C_x)=0$.
It implies that the  maximization can be restricted to the intersection of two cones coming out of $C_p^{-1}$ and $C_x$ in the Minkowski space with the coordinate being the coefficients of Pauli expansion and the norm being the determinant. 
Then, it can be expressed as the maximization over $\theta$ for the following function~\cite{Adesso2005};
\begin{align*}
 m(\theta) = & 1 + \left(-k_p - k_p^2 k_x + k_x na b + \sqrt{(a + k_p^2 b - a b^2) (b + a (k_p^2 - a b))}\cos\theta\right)^2\\
 &\times \left[2 (k_p^2 - a b) \left(-2 k_p k_x - a^2 - b^2 + \frac{\left\{2 k_p^3 a b - k_x a b (-2 + a^2 + b^2) + k_p^2 k_x (a^2 + b^2) + k_p (b^2 + a^2 (1 - 2 b^2))\right\} \cos\theta}{\sqrt{(a + k_p^2 b - a b^2) (b + a (k_p^2 - a b))}}\right. \right.\\ 
 & \left. \left. - (a^2 - b^2) \sqrt{1 - \frac{(k_p - k_p^2 k_x + k_x a b)^2}{(a + k_p^2 b - a b^2) (b + a (k_p^2 - a b))}}\sin\theta \right)\right]^{-1}.
\end{align*}

We now consider GLEMS. It is the class of two-mode states that are marginals of three-mode pure states.
Suppose, without loss of generality, the QCM for the given state in GLEMS is represented by the standard form $V_{AB} = \begin{pmatrix}
  a & k_x \\
  k_x & b
 \end{pmatrix}\oplus
 \begin{pmatrix}
  a & k_p \\
  k_p & b
 \end{pmatrix}$
with $k_x\geq |k_p|$.
Then, again without loss of generality, QCM for a three-mode pure state purifying $V_{AB}$ can be written in the standard form 
\bbb
 V_{ABC}=\begin{pmatrix}
  a & k_x & k_x' \\
  k_x & b & k_x'' \\
  k_x' & k_x'' & c
 \end{pmatrix}\oplus
 \begin{pmatrix}
  a & k_p & k_p' \\
  k_p & b & k_p'' \\
  k_p' & k_p'' & c
 \end{pmatrix}
\eee
because any three-mode pure state can be brought into this standard form by local symplectic transformations~\cite{Adesso2006three}. 
Observe that the Gaussian entanglement of assistance for the two-mode state $V_{AB}$ is equivalent to the Gaussian localizable entanglement~\cite{Fiurasek2007} induced from $V_{ABC}$ by making a measurement on $C$.
In~\cite{Fiurasek2007}, it was shown that the optimal Gaussian measurement on $C$ for three-mode pure states in the standard form is always taken to be a homodyne measurement projecting onto the eigenstates of the quadrature $x$ or $p$, which keeps the post-measurement state on $AB$ in $x-p$ separate form. Thus, the optimization trick described above can be applied for GLEMS. 
Note that in the case of GLEMS, $k_x$ and $k_p$ are completely determined by $a$, $b$, and (inverse of) global purity $g\coloneqq (\Tr[\rho_{AB}^2])^{-1}=\sqrt{(ab-k_x^2)(ab-k_p^2)}$ as
\ba
 k_x = \frac{1}{4\sqrt{ab}}\left(\sqrt{\left[(a-b)^2-(g+1)^2\right]\left[(a-b)^2-(g-1)^2\right]}+\sqrt{\left[(a+b)^2-(g+1)^2\right]\left[(a+b)^2-(g-1)^2\right]}\right)
\ea
\ba
 k_p = \frac{1}{4\sqrt{ab}}\left(\sqrt{\left[(a-b)^2-(g+1)^2\right]\left[(a-b)^2-(g-1)^2\right]}-\sqrt{\left[(a+b)^2-(g+1)^2\right]\left[(a+b)^2-(g-1)^2\right]}\right)
\ea
In~\cite{Adesso2005}, it was observed that in this case the term with $\sin\theta$ in $m(\theta)$ vanishes, and it reduces to the simpler form
\ba
 m(\theta)_{\rm GLEMS} = 1 + \frac{(A\cos\theta+B)^2}{2(ab-k_p^2)\left[(g^2-1)\cos\theta + g^2 +1\right]}
\ea
where $A=k_x(ab-k_p^2)+k_p$, $B=k_x(ab-k_p^2)-k_p$.
Straightforward calculation yields that $\partial_\theta m(\theta) = 0$ is realized for $\theta = 0, \pi, \pm \theta^*$ where 
\bbb
\pm\theta^*=\arccos{\left[\frac{3+g^2}{1-g^2}-\frac{2k_p}{k_x(ab-k_p^2)+k_p}\right]}.
\eee
Observe that $ab-k_p^2>0$ due to the positivity of $V_{AB}$, $ab-k_p^2 \geq \sqrt{ab-k_x^2}\sqrt{ab-k_p^2} = g$ because $k_x\geq |k_p|$, and $g\geq 1$. 
One can then obtain 
\bb
 m(0)-m(\pi)=\frac{-k_p^2+\frac{k_x^2(ab-k_p^2)^2}{g^2}}{ab-k_p^2}\geq 0,
\ee
\bb
 m(0)-m(\pm\theta^*)=\frac{\left[abk_x(1+g^2)-k_p(-2g^2+k_xk_p(1+g^2))\right]^2}{(ab-k_p^2)g^2(-1+g^2)^2}\geq 0.
\ee
Thus, $\theta=0$ achieves the global maximum and we obtain for GLEMS
\bb
E_a^{\G}(V_{AB}) = f\left[1 + \frac{k_x^2}{ab-k_x^2}\right]. 
\ee
\end{proof}

\section{Proof of Theorem~\ref{Ea id thm}}

Here we provide a complete proof of Theorem~\ref{Ea id thm}, which establishes an additive upper bound on $E_a^\G$ (and thus on $E_a^{\G,\infty}$) for all monotones $E$ that derive from a concave function $f$ via~\eqref{characterization E}. Before we delve into the proof, let us fix some notation and establish some preliminary results. Any function $f:\R\to \R$ can be extended to $N\times N$ real symmetric matrices via spectral calculus, i.e., by setting $f(A)\coloneqq \sum_i f(\lambda_i)$, where $\lambda_1,\ldots, \lambda_N$ are the eigenvalues of $A$. If $N=2n$ is even and $A>0$ is strictly positive definite, it is also possible to consider the \emph{symplectic extension}
\bb
F(A)\coloneqq \sum_{i=1}^n f(\nu_i)\, ,
\label{symplectic extension}
\ee
where $\nu_1,\ldots, \nu_n$ are the symplectic eigenvalues of $A$. The following lemma, whose proof is inspired by the techniques in~\cite{bhatia15}, presents a remarkable variational formula for $F$ in the relevant case when $f$ is concave and monotonically non-decreasing.

\begin{lemma} \label{F variational lemma}
Let $f:\R_+\to \R$ be a concave non-decreasing function, and consider its symplectic extension $F$ as defined by~\eqref{symplectic extension}. Then one has that
\bb
F(A) = \min_{M\in \mathrm{Sp}(2n)} f\left( \Delta_s\left( M A M^\intercal\right) \right) ,
\label{F variational}
\ee
where $\mathrm{Sp}(2n)$ denotes the symplectic group, and $\Delta_s\left( \left(\begin{smallmatrix} X & Z \\ Z^\intercal & P \end{smallmatrix}\right) \right)\coloneqq \frac12 \sum_{i=1}^n \left( X_{ii}+P_{ii}\right) \ketbra{i}$.
\end{lemma}

\begin{proof}
This argument is inspired by the techniques introduced by Bhatia and Jain in~\cite{bhatia15}. First of all, take as $M$ the symplectic matrix $M_0$ that brings $A$ in Williamson's form~\cite{willy}. Since $M_0 A M_0^\intercal = \left( \begin{smallmatrix} D & 0 \\ 0 & D \end{smallmatrix} \right)$ for a diagonal $D=\sum_{i=1}^n \nu_i \ketbra{i}$, one has that
\bbb
\inf_{M\in \mathrm{Sp}(2n)} f\left( \Delta_s\left( M A M^\intercal\right) \right) \leq f\left( \Delta_s\left( M_0 A M_0^\intercal\right) \right) = f(D) = F(A)\, .
\eee

It remains to prove the reverse inequality. In what follows, we will employ few concepts from the theory of majorization; for an introduction, we refer the reader to the excellent monograph~\cite[Chap.~1]{MARSHALL-OLKIN}.  Since in the r.h.s.\ of~\eqref{F variational} one is anyway optimising over all symplectic matrices, we can without loss of generality assume that $A$ is in Williamson's form, i.e., that $A=\left( \begin{smallmatrix} D & 0 \\ 0 & D\end{smallmatrix} \right)$ with $D=\sum_{i=1}^n \nu_i \ketbra{i}$. Now, consider a symplectic matrix partitioned as $M=\left( \begin{smallmatrix} P & Q \\ R & S \end{smallmatrix} \right)$. Observe that
\begin{align*}
    \Delta_s\left( M A M^\intercal \right) &= \Delta_s\left(\left( \begin{smallmatrix} PDP^\intercal + QDQ^\intercal & PDR^\intercal QDS^\intercal \\ RDP^\intercal + SDQ^\intercal & RDR^\intercal + SDS^\intercal \end{smallmatrix} \right) \right) \\
    &= \frac12 \sum_{i=1}^n \left( PDP^\intercal + QDQ^\intercal + RDR^\intercal + SDS^\intercal \right)_{ii} \ketbra{i} \\
    &= \frac12 \sum_{i=1}^n \sum_{j=1}^n \left( P_{ij}^2 + Q_{ij}^2 + R_{ij}^2 + S_{ij}^2 \right) \nu_j \ketbra{i} \\
    &= \sum_{i=1}^n \left( \widetilde{M} \nu \right)_i \ketbra{i}\, ,
\end{align*}
where we introduced the $n\times n$ matrix $\widetilde{M}$ defined by
\bbb
\widetilde{M}_{ij} \coloneqq \frac12 \left( P_{ij}^2 + Q_{ij}^2 + R_{ij}^2 + S_{ij}^2 \right)
\eee
as well as the shorthand $\nu\coloneqq (\nu_1,\ldots, \nu_n)^\intercal$ for the column vector of symplectic eigenvalues. We then see that the entries of $\Delta_s(MAM^\intercal)$ are obtained by applying the matrix $\widetilde{M}$ to $\nu$, in formula 
\bbb
\mathrm{diag}\left( \Delta_s (MAM^\intercal)\right) = \widetilde{M}\nu\, .
\eee
It is shown in~\cite[Theorem~6]{bhatia15} that when $M$ is symplectic the matrix $\widetilde{M}$ is doubly superstochastic, which via~\cite[Proposition~2.D.2.b]{MARSHALL-OLKIN} implies that
\bbb
\mathrm{diag}\left( \Delta_s (MAM^\intercal)\right) \prec^w \nu\, ,
\eee
with $\prec^w$ denoting weak supermajorization. Using~\cite[3.C.1.b]{MARSHALL-OLKIN} (with $g=-f$, which is convex and non-increasing), we then see that
\bbb
f\left( \Delta_s (MAM^\intercal)\right) = \sum_{i=1}^n f\left( \Delta_s (MAM^\intercal)_{ii} \right) \geq \sum_{i=1}^n f(\nu_i) = F(A)\, .
\eee
Since $M$ was an arbitrary symplectic matrix, this completes the proof.
\end{proof}

\begin{cor} \label{F concavity cor}
Let $f:\R_+\to \R$ be a concave non-decreasing function. Then its symplectic extension $F$ defined by~\eqref{symplectic extension} is monotonically non-decreasing and concave on the set of (strictly) positive definite matrices.
\end{cor}

\begin{proof}
The fact that $F$ is a monotonic function is very clear from~\eqref{F variational}, and is also a consequence of the monotonicity principle for symplectic eigenvalues established in~\cite{giedkemode}. We now move on to the proof of concavity. Let $A,B>0$ be strictly positive definite. Take a symplectic matrix $M$ that achieves the minimum in~\eqref{F variational} for $(A+B)/2$ (e.g., the one that brings $(A+B)/2$ into Williamson's form). Then, using the concavity of $f$ we obtain that
\begin{align*}
    F\left(\frac{A+B}{2}\right) &= f\left( \Delta_s \left( M\, \frac{A+B}{2}\, M^\intercal \right)\right) \\
    &= f\left( \frac12 \Delta_s \left( M A M^\intercal \right) + \frac12 \Delta_s \left( M B M^\intercal \right)\right) \\
    &= \sum_{i=1}^n f\left( \frac12 \Delta_s \left( M A M^\intercal \right)_{ii} + \frac12 \Delta_s \left( M B M^\intercal \right)_{ii} \right) \\
    &\geq \sum_{i=1}^n \left\{ \frac12 f\left( \Delta_s \left( M A M^\intercal \right)_{ii}\right) + \frac12 f\left(\Delta_s \left( M B M^\intercal \right)_{ii} \right)\right\} \\
    &= \frac12 f\left( \Delta_s\left( M A M^\intercal\right) \right) + \frac12 f\left( \Delta_s\left( M B M^\intercal\right) \right) \\
    &\geq \frac12 F(A) +\frac12 F(B)\, , 
\end{align*}
as claimed.
\end{proof}

We are now ready to give a full proof of Theorem~\ref{Ea id thm}.

\begin{proof}[Proof of Theorem~\ref{Ea id thm}]
We start by proving~\eqref{upper bound Ea}. We only have to show that given any QCM $V_{AB}$ and any pure QCM $\tau_{AB}\leq V_{AB}$, one has that
\bb
E(\tau_{AB}) = F(\tau_{A}) \leq n\, f\left( \frac{\|V_{AB}\|_\infty^2+1}{2\|V_{AB}\|_\infty} \right) .
\label{upper bound F}
\ee
In fact, the bound $E_a^\G(V_{AB}) \leq n\, f\left( \frac{\|V_{AB}\|_\infty^2+1}{2\|V_{AB}\|_\infty} \right)$ will then follow by taking the supremum over $\tau_{AB}$. The regularized bound on $E_{a}^{\G,\infty}$ is easily deduced as the r.h.s.\ of~\eqref{upper bound F} does not change when $V_{AB}$ is replaced by $V_{AB}^{\oplus\ell}$.

We now prove the inequality in~\eqref{upper bound F}. We can assume without loss of generality that $n=n_A\leq n_B$. Set $t\coloneqq \lambda_{\max}(\tau) = \|\tau\|_\infty$, so that $t\leq \|V_{AB}\|_\infty$. Here we used the fact that for positive matrices the operator norm coincides with the maximal eigenvalue, denoted with $\lambda_{\max}$. Let us write
\begin{align*}
n\, f\left( \frac{\|V_{AB}\|_\infty^2+1}{2\|V_{AB}\|_\infty} \right) &\textgeq{1} n\, f\left( \frac{t^2+1}{2t} \right) \\
&= n\, f\left( \frac{t+t^{-1}}{2} \right) \\
&\texteq{2} n\, f\left( \lambda_{\max}\left( \frac{\tau+\tau^{-1}}{2} \right)\right) \\
&\texteq{3} n\, f\left( \lambda_{\max}\left( \frac{\tau+\Omega \tau \Omega^\intercal}{2} \right)\right) \\
&\textgeq{4} n\, f\left( \lambda_{\max}\left( \left(\frac{\tau+\Omega \tau \Omega^\intercal}{2}\right)_A \right)\right) \\
&\textgeq{5} n\, f\left( \nu_{\max}\left( \left(\frac{\tau+\Omega \tau \Omega^\intercal}{2}\right)_A \right)\right) \\
&\textgeq{6} F \left( \left(\frac{\tau+\Omega \tau \Omega^\intercal}{2}\right)_A \right) \\
&= F \left( \frac{\tau_A+\Omega_A \tau_A \Omega_A^\intercal}{2} \right) \\
&\textgeq{7} \frac12 F( \tau_A) + \frac12 F(\Omega_A\tau_A \Omega_A^\intercal) \\
&\texteq{8} F(\tau_A)\, .
\end{align*}
The above steps are justified as follows. 1: Descends from the fact that $t\mapsto (t^2 +1)/(2t)$ is an increasing function on $[1,\infty)$, and indeed $t\geq 1$ because $\det\tau = 1$. 2: Since $\tau$ and $\tau^{-1}$ commute, the spectrum of $\tau+\tau^{-1}$ has the form $\{\lambda_{i}+\lambda_i^{-1}\}_i$, with $\{\lambda_i\}_i$ being the spectrum of $\tau$; its maximum is thus achieved for $\lambda_i=\lambda_{\max}(\tau)=t$. 3: Comes from the observation that $\tau$ is a symplectic matrix, hence $\tau \Omega \tau = \Omega$ and therefore $\tau^{-1} = \Omega \tau \Omega^\intercal$. 4: Amounts to noting that the maximal eigenvalues never increases when passing from a matrix to one of its sub-blocks, see e.g.~\cite[Lemma~3.3.1]{HJ2}; remember also that $f$ is non-decreasing. 5: Here $\nu_{\max}$ denotes the maximal symplectic eigenvalue, which is never larger than the maximal eigenvalue by~\cite[Theorem~11]{bhatia15}. 6: Is a direct consequence of~\eqref{symplectic extension}. 7: Is an application of Corollary~\ref{F concavity cor}, which establishes the concavity of $F$. 8: Derives from the elementary observation that $\Omega_A$ itself is a symplectic matrix, hence the symplectic eigenvalues of $\tau_A$ and $\Omega_A \tau_A \Omega_A^\intercal$ are the same. This concludes the proof of~\eqref{upper bound Ea}.

To see why~\eqref{Ea id} holds, start by observing that from~\eqref{upper bound Ea} we trivially deduce that $E_a^{\G,\infty}(k\id_{AB})\leq n f\left( \frac{k^2+1}{2k} \right)$. In order to establish the converse, assume without loss of generality that $n=n_A\leq n_B$, and write formally $k \id_{AB} = \bigoplus_{j=1}^n \left(k\id_{A_j B_j}\right) \oplus \left( k \id_{B_{n+1}\ldots B_{n_B}}\right)$, where $A_j$ is the $j$-th mode on $A$, and analogously for $B$. We then have that
\begin{align*}
    E_a^\G (k\id_{AB}) &= E_a^\G \left(\bigoplus_{j=1}^n \left(k\id_{A_j B_j}\right) \oplus \left( k \id_{B_{n+1}\ldots B_{n_B}}\right) \right) \\
    &\textgeq{9} \sum_{j=1}^n E_a^\G\left( k \id_{A_j B_j} \right)  \\
    &\texteq{10} n f\left( \frac{k^2+1}{2k} \right) .
\end{align*}
Here, step 9 comes from superadditivity of $E_a^\G$, while step 10 is an application of the two-mode formula~\eqref{entanglement product}. This proves~\eqref{Ea id}.
\end{proof}

\bibliographystyle{apsrmp4-2}
\bibliography{biblio}
\end{document}